\documentclass[11pt]{llncs}
\usepackage{graphicx,graphics,hyperref,color}
\usepackage{alphabeta}
\usepackage{float}
\usepackage{amsmath}
\usepackage[ruled,vlined,linesnumbered,longend]{algorithm2e}
\usepackage{algorithmic}

\newlength\myindent
\usepackage{eqparbox}

\usepackage[margin=1in]{geometry}
\usepackage{thmtools}
\usepackage{preamble}
\usepackage{maths}

\newcommand{\vmax}{v_{\text{max}}}

\def\alg{\mathrm{ALG}}

\def\opt{\mathrm{OPT}}
\def\Exp{\mathbb{E}}

\allowdisplaybreaks 

\title{A Competitive Posted-Price Mechanism for Online Budget-Feasible Auctions}
\institute{School of ECE, National Technical University of Athens, Greece \and Université Paris-Dauphine, Université PSL, CNRS \and
 Archimedes, Athena Research Center, Greece\\\email{andcharalamp@gmail.com, fotakis@cs.ntua.gr, thanostolias@mail.ntua.gr, panagiotis.patsilinakos@dauphine.psl.eu}}
\author{Andreas Charalampopoulos\inst{1,3} \and Dimitris Fotakis\inst{1,3} \and Panagiotis Patsilinakos\inst{2} \and Thanos Tolias\inst{1,3} } 

\date{19/12/2024}

\begin{document}

\maketitle

\begin{abstract}
We consider online procurement auctions, where the agents arrive sequentially, in random order, and have private costs for their services. The buyer aims to maximize a monotone submodular value function for the subset of agents whose services are procured, subject to a budget constraint on their payments. We consider a posted-price setting where upon each agent's arrival, the buyer decides on a payment offered to them. The agent accepts or rejects the offer, depending on whether the payment exceeds their cost, without revealing any other information about their private costs whatsoever. We present a randomized online posted-price mechanism with constant competitive ratio, thus resolving the main open question of (Badanidiyuru, Kleinberg and Singer, EC 2012). Posted-price mechanisms for online procurement typically operate by learning an estimation of the optimal value, denoted as $\opt$, and using it to determine the payments offered to the agents. The main challenge is to learn $\opt$ within a constant factor from the agents' accept / reject responses to the payments offered. Our approach is based on an online test of whether our estimation is too low compared against $\opt$ and a carefully designed adaptive search that gradually refines our estimation. 
\end{abstract}

\def\N{\mathcal{N}}
\def\reals{\mathbb{R}}
\def\nats{\mathbb{N}}
\def\eps{\varepsilon}

\section{Introduction}
\label{s:intro}

We consider a model of procurement auctions, introduced by Singer \cite{Singer10}, where a set $\N$ of $n$ strategic agents are willing to offer their services to a buyer. Each agent $i \in \N$ has a private cost $c_i \in \reals_{\geq 0}$ for offering their service. The buyer aims to maximize a public monotone submodular value function $f : 2^{\N} \to \reals_{\geq 0}$ for the subset of agents hired, subject to a hard budget constraint that the total amount of payments to the agents should not exceed the buyer's budget $B$. 

In the algorithmic problem, the buyer is aware of the agent costs $c_1, \ldots, c_n$ and aims to compute a subset $S \subseteq \N$ of agents (a.k.a. a solution) that maximizes $f(S)$ subject to $\sum_{i \in S} c_i \leq B$. Maximizing a monotone submodular function subject to a budget (a.k.a.  knapsack) constraint is a classical NP-hard optimization problem whose approximability is well understood (see e.g., \cite{NemhauserWF78,KhullerMN99,Sviri04,BadaDO19}). 

In the mechanism design setting of \cite{Singer10,Singer13}, we aim to design \emph{individually rational} and \emph{truthful} (a.k.a. \emph{incentive compatible}) mechanisms that compensate each agent $i$ included in the solution with a carefully chosen payment $p_i$. Individual rationality requires that $p_i \geq c_i$, i.e., that the payment allocated to each agent in the solution should be no less than their cost. Truthfulness requires that the agent selection and the corresponding payments should ensure that no agent has an incentive to falsely report a higher cost to the mechanism, in an attempt to extract a larger payment. 
The mechanism must be budget-feasible with respect to the payments $p_i$ to the agents included in the solution, i.e., it must be $\sum_{i\in S} p_i \leq B$, and in the ideal case, the mechanism should also be computationally efficient.   
Conceptually, individual rationality and truthfulness require that the payments are large enough to cover the agent costs and incentivize truthful reporting, while budget feasibility requires that the payments are kept low, so that the value of the solution is maximized. 



Compensating strategic agents for their effort in a way that is incentive-compatible, value-efficient and budget-feasible is an intriguing algorithmic problem with significant practical applications, which include crowdsourcing \cite{SingerM13,SinglaK13,Liu2015}, crowdsensing \cite{Zheng2022}, participatory sensing \cite{Restuccia2016} (see also the references in the survey of Liu et al. \cite{LCLW2024_survey}).

Since its introduction by Singer \cite{Singer10} and due to its practical and theoretical significance, the problem of designing truthful mechanisms for budget-feasible procurement auctions has received significant attention. A long line of research has shown that budget-feasible procurement can be approximated within small constant factors by computationally efficient truthful mechanisms for additive \cite{GravinJLZ2020}, monotone submodular \cite{JalalyT2021,BalkanskiGGST2022} and non-monotone submodular valuations \cite{AmanatidisKS22,BalkanskiGGST2022} (see Section~\ref{sec:related} for a more detailed discussion), in the \emph{bidding} (a.k.a. \emph{sealed-bid}) setting, where the mechanism may ask the agents to report their costs. 

Previous work also considered \emph{online} budget-feasible procurement auctions with \emph{secretary arrivals}, where the agents are chosen adversarially and arrive sequentially in random order, and the mechanism has to decide irrevocably about selecting an agent in the solution and the corresponding payment, upon the agent's arrival and without any knowledge about future agent arrivals. With agent bidding,  online procurement auctions admit constant factor approximations for additive \cite{SingerM13}, monotone submodular \cite{Bada2012} and non-monotone submodular \cite{AmanatidisKS22} valuations. 

However, in practical applications, such as crowdsourcing, in addition to dealing with the agents online, it is desirable (and often necessary) that the mechanism interacts with the agents through \emph{sequential posted pricing}. 
A posted-price mechanism decides on the payment $p_i$ offered to an agent $i$ upon $i$'s arrival. Then, the agent accepts or rejects the offer, depending on whether $p_i \geq c_i$ or not, without revealing any other information about their private cost $c_i$ whatsoever. 
Posted-price mechanisms are used extensively in practice, due to their simplicity, transparency, obvious truthfulness \cite{Li2017} and agent privacy regarding their costs. E.g., posted pricing is the standard form of pricing in crowdsourcing applications, such as Mechanical Turk (see also the discussion in \cite{BalkHart16,ChawlaHMS2010} for the motivation and the advantages of sequential posted-price mechanisms). 

%
%

In the context of online budget-feasible procurement, Badanidiyuru, Kleinberg and Singer \cite{Bada2012} asked whether \emph{posted-price mechanisms are as powerful, in terms of the asymptotic behavior of their competitive ratio, as standard bidding mechanisms}. Towards resolving this question, they presented a randomized constant-competitive posted-price mechanism for symmetric submodular valuations%
\footnote{A function $f : 2^{\N} \to \reals_{\geq 0}$ is \emph{symmetric submodular} if there is a nondecreasing concave function $g:\nats \to \reals_{\geq 0}$ such that $f(S) = g(|S|)$ for all $S \subseteq \N$.}
and agents sampled independently from an unknown distribution. For monotone submodular valuations and secretary arrivals, Badanidiyuru et al. \cite{Bada2012} presented a randomized constant-competitive  online mechanism with agent bidding, and a randomized $O(\log n)$-competitive posted-price mechanism. Their main open question concerned the existence of a constant-competitive posted-price mechanism for monotone submodular valuations and secretary agent arrivals. 

Subsequently, Balkanski and Hartline \cite{BalkHart16} considered posted-price mechanisms in the Bayesian setting, where the agents arrive as independent samples from a known distribution. They presented mechanisms that are constant-competitive against the Bayesian optimal mechanism for additive, symmetric submodular and monotone submodular value functions. 

\subsection{Contribution and General Approach}
\label{sec:approach}

In this work, we present a randomized constant-competitive posted-price mechanism for online procurement auctions with monotone submodular valuations, thus resolving the main open question of Badanidiyuru et al. \cite{Bada2012}. More formally, we show the following:

\begin{theorem}\label{thm:main_informal}
There is a universally truthful $O(1)$-competitive randomized posted-price mechanism for budget-feasible online procurement auctions with secretary agent arrivals and  monotone submodular buyer's valuations.
\end{theorem}

The mechanism of Theorem~\ref{thm:main_informal} satisfies budget-feasibility with certainty, and is individually rational and universally (and obviously) truthful%
\footnote{A randomized mechanism is universally truthful if it can be expressed as a probability distribution over truthful deterministic mechanisms. A mechanism is \emph{obviously truthful} \cite{Li2017} if it has an equilibrium in obviously dominant strategies. A strategy is \emph{obviously dominant} if the best outcome under any possible deviation from it is no better than the worst outcome under the given strategy.}, 
with the latter properties following directly from the definition of posted-price mechanisms, where the agents to act in their best interest when selecting whether to accept the payment offered or not. 

Posted-price mechanisms for online procurement typically operate by learning an estimation of the optimal value $\opt$ (usually referred to as a \emph{threshold}, for brevity) $\hat{t}$ and using it to determine the payments (a.k.a. \emph{prices}) offered to the agents. We restrict our attention to \emph{linear-price} mechanisms, where the price offered to agent $i$ is $p_i = f_S(i) B / \hat{t}$, where $B$ is the buyer's budget, $\hat{t}$ is the mechanism's current threshold, $S$ is the subset of agents selected in the solution up to $i$'s arrival, and $f_S(i) = f(S \cup \{ i \}) - f(S)$ is the marginal value due to $i$'s inclusion in the solution.

One can show (see e.g., the proof of Lemma~\ref{lem:threshold} or \cite[Lemma~4.4]{Bada2012}) that linear-price mechanisms with thresholds $\hat{t}$ that remain within a constant fraction of $\opt$ for a constant fraction of the agent sequence are constant-competitive. Hence, the main challenge in the design of efficient linear-price mechanisms (and the main difference between them) is how to maintain a threshold $\hat{t}$ that mostly remains within a constant fraction of $\opt$. 

The posted-price mechanism of Badanidiyuru et al. \cite[Section~4]{Bada2012} uses linear prices and a fixed threshold $\hat{t}$ determined before the first offer to an agent. In Theorem~\ref{thm:non-adaptive-lower-bound}, we show that the competitive ratio of any linear-price mechanism that selects a fixed threshold without any knowledge of the agent costs is $\Omega(\log n)$. Hence, we consider \emph{adaptive linear-price} mechanisms based on a threshold that evolves over time. 

Our approach is inspired by the constant-competitive adaptive linear-price mechanism of \cite[Section~3]{Bada2012} for symmetric submodular valuations and agents sampled independently from an unknown distribution.
As in \cite[Section~3]{Bada2012}, our approach is based on an online test (cf. Mechanism~\ref{alg:TestTHRESHOLD}, usually referred to as \emph{TestThreshold}) that efficiently decides if the current threshold $\hat{t}$ is less than $\opt / 4$. TestThreshold comes with an one-sided guarantee: with constant probability, TestThreshold succeeds (i.e., it responds that the current threshold $\hat{t} \leq \opt/4$), if it is applied with any threshold $\hat{t} \leq \opt / 4$ to a random subset of agents of size $\Omega( \vmax n / \opt)$, where $\vmax = \max_{i \in \N} \{ f(\{i\}) \}$ is the maximum value of an agent (Property~\ref{property:b}). If TestThreshold is applied with a threshold $\hat{t} > \opt / 4$, there is no guarantee on its response. If TestThreshold fails (i.e., it responds that the current threshold $\hat{t} > \opt / 4$), we become aware that $\hat{t}$ is rather too large and should be decreased; otherwise, TestThreshold contributes $\Omega(\vmax \hat{t} / OPT)$ to the total value of the mechanism.  

The high level idea of our approach is to first obtain a rough estimation of $\opt$ and then apply binary search, guided by TestThreshold, to the possible range of $\opt$ defined by our initial estimation. At the conceptual level, our approach consists of three periods:
\begin{description}
\item[Learning $\vmax$:] 
As in standard secretary algorithms, we start with a \emph{learning period} that considers a constant fraction of the agents, where we do not make any offers and aim to estimate $\vmax$ with constant probability. We note that  $\vmax \leq \opt \leq n \cdot \vmax$.

\item[Binary Search:] 
Using TestThreshold, we apply binary search to the set $\{ \vmax, 2\vmax, 4\vmax, \ldots$, $2^{\log n} \vmax\}$ of different estimations of $\opt$ expressed as a power of $2$ times $\vmax$. This period aims to compute a threshold $\hat{t}$ within a constant fraction of $\opt$. Our application of binary search requires $\Theta(\log\log n)$ \emph{phases}, each consisting of $O(\log\log n)$ negatively-dependent \emph{rounds} where we apply TestThreshold (Appendix~
\ref{negdep}). Hence, we can prove that using $O((\log\log n)^2)$ applications of TestThreshold to the next constant fraction of agents, we end up with a threshold $\hat{t} \geq \opt/8$ with constant probability (Lemma~\ref{lemma:dist}). 

\item[Exploitation:]
Applying TestThreshold to the last constant fraction of agents, where every time the test succeeds $\hat{t}$ is doubled, and every time the test fails $\hat{t}$ is halved, we collect an expected total value that is within a constant fraction of $\opt$ (Lemma~\ref{lemma:compratio}). 
\end{description}

The three-period approach above comprises a constant-competitive posted-price mechanism for online budget-feasible procurement, but only under a \emph{large market assumption}. Specifically, in order to show a constant competitive ratio, we need to assume that $\opt / \vmax = \Omega((\log\log n)^2)$, i.e., that the optimal solution consists of sufficiently many agents.
Such large market is necessary in order to guarantee that each of the $O((\log\log n)^2)$ applications of TestThreshold considers a random subset of agents of size $\Omega( \vmax n / \opt)$. This is necessary, so that we can prove that the entire period of Binary Search succeeds with constant probability. 

We can deal with the case where $\opt / \vmax = o((\log\log n)^2)$ by applying Dynkin's secretary algorithm \cite{Dynkin63} with constant probability, which results in competitive ratio of $O((\log\log n)^2)$ (which can be improved to $O(\log\log n)$ with a careful analysis). 

To remove the large market assumption and obtain a constant-competitive mechanism, we assume that $\opt \geq \beta \vmax$, where $\beta$ is a sufficiently large constant. We define a power tower sequence as $\psi_1 = \beta$ and $\psi_{j+1} = 2^{\psi_j}\cdot \psi_j$ and apply TestThreshold with thresholds $\psi_j \vmax$ in order to determine a pair of consecutive intervals $[\psi_j \vmax, \psi_{j+1} \vmax)$ and $[\psi_{j+1} \vmax, \psi_{j+2}\vmax)$ such that $\opt$ is included in their union. We show how to compute such a pair of intervals with constant probability by $O(\log^\ast n)$ applications of TestThreshold to a constant fraction of the agent sequence in total (Lemma~\ref{lemma:Fit}). 
Selecting one of these intervals $[\psi_j \vmax, \psi_{j+1} \vmax)$ and $[\psi_{j+1} \vmax, \psi_{j+2}\vmax)$ at random and applying Binary Search and Exploitation, as described above, to the chosen interval (instead of the interval $[\beta \vmax, n \vmax)]$ initially considered) satisfies the required large market property and results in a constant-competitive posted-price mechanism, are required by 
Theorem~\ref{thm:main_informal}. 

\subsection{Related Work}
\label{sec:related}

From an algorithmic viewpoint, the greedy algorithm of Nemhauser et al. \cite{NemhauserWF78} gives an approximation ratio of $e/(e-1)$ for the algorithmic problem of maximizing a monotone submodular function subject to a knapsack constraint \cite{Sviri04}, and this approximation guarantee is best possible in polynomial-time  \cite{KhullerMN99}, under standard complexity assumptions. More recently, Badanidiyuru et al. \cite{BadaDO19} presented a $(9/8+\eps)$-approximation with  polynomially many demand queries.

Singer \cite{Singer10,Singer13} was the first to study the approximability of budget-feasible procurement auctions by truthful mechanisms and presented a randomized $O(1)$-approximation for monotone submodular valuations. Subsequently, Chen et al. \cite{ChenGL2011} significantly improved the approximation ratio to $7.91$ (resp. $3$) for randomized and $8.34$ (resp. $2+\sqrt{2}$) for deterministic mechanisms when the buyer's valuation is monotone submodular (resp. additive), using a greedy (resp. knapsack) based strategy. Chen et al. also proved an unconditional lower bound of $1+\sqrt{2}$ (resp. $2$) on the approximability of budget-feasible auctions with additive values by deterministic truthful (resp. randomized universally truthful) mechanisms. For additive values, Gravin et al. \cite{GravinJLZ2020} matched the best possible approximation ratio of $2$ for randomized mechanisms and presented a $3$-approximate deterministic mechanism. Anari et al. \cite{AnariGN2014} gave best possible $e/(e-1)$-approximation mechanisms for large markets. 

For monotone submodular valuations, Jalaly and Tardos \cite{JalalyT2021} presented a randomized polynomial-time $5$-approximation mechanism, thus improving on \cite{ChenGL2011}, 
and a deterministic (resp. randomized) $4.56$ (resp. $4$) approximation mechanism, which assumes access to an exact algorithm for value maximization.
The approximation ratio for large markets was improved to $2$ for deterministic (possibly exponential-time) mechanisms and to $3$ for randomized polynomial-time mechanisms by Anari et al. \cite{AnariGN2014}. Recently, Balkanksi et al. \cite{BalkanskiGGST2022} presented a deterministic polynomial-time clock auction that is $4.75$-approximate for monotone submodular valuations. The best known approximation ratio for monotone submodular valuations is $4.45$ for deterministic and $4.3$ for randomized mechanisms, achieved by the polynomial-time clock auction of Han et al. \cite{HanWHC23}. 

Polynomial-time constant-factor approximations are also known for non-monotone submodular valuations \cite{AmanatidisKS22,HuangHC023}, where the best known ratio is $64$ for deterministic  \cite{BalkanskiGGST2022} and $12$ for randomized mechanisms \cite{HanWHC23} (both achieved by clock auctions). Moreover, constant-factor approximation mechanisms (albeit non-polynomial time ones) are known for XOS valuations \cite{BeiCGL2012,AmanatidisBM17}, while for subadditive valuations the best known approximation guarantee is $O(\frac{\log n}{\log\log n})$ due to Balkanski et al. \cite{BalkanskiGGST2022}, which is achieved by a clock auction, matches the randomized approximation of \cite{BeiCGL2012} and improves on the $O(\log^3 n)$ deterministic approximation of \cite{DobzinskiPS11}. We refer an interested reader to the survey of Liu et al. \cite{LCLW2024_survey} for a detailed discussion of previous work on budget-feasible procurement. 

The online version of budget-feasible auctions, where the agents arrive sequentially in random order and the decision about their acceptance in the solution is online and irrevocable, was introduced by Singer and Mittal \cite{SingerM13} and Badanidiyuru et al. \cite{Bada2012}. 
Amanatidis et al. \cite{AmanatidisKS22} presented $O(1)$-competitive randomized online universally truthful mechanisms for monotone and non-monotone submodular valuations in the bidding setting. Online budget-feasible procurement is closely related to the \emph{submodular knapsack secretary} problem \cite{BateniHZ13,KesselheimT17}, where budget feasibility is with respect to the agent costs (which are revealed to the algorithm upon each agent's arrival). Feldman et al. \cite{FeldmanNS11} presented a randomized $20e$-approximation for submodular knapsack secretaries.

The clock auctions of Balkanksi et al. \cite{BalkanskiGGST2022} and Han et al. \cite{HanWHC23} achieve small constant approximation guarantees without resorting to  bidding (i.e., the agents never report their costs to the mechanism), but
they are not online, because the agents receive multiple offers by the mechanism. To the best of our knowledge, the online $O(\log n)$-competitive mechanism of \cite[Section~4]{Bada2012} and the constant-competitive mechanisms in the Bayesian setting of 
\cite{BalkHart16} are the only known posted-price mechanisms for online budget-feasible procurement auctions with monotone submodular valuations. 

A quite standard approach to the design of efficient budget-feasible mechanisms (see e.g., \cite[Section~5]{Bada2012} and in 
\cite[Section~4]{AmanatidisKS22}) and of online algorithms for submodular knapsack secretaries (see e.g., \cite{FeldmanNS11}) is to first learn the costs of a random subset of agents. Then, based on these costs, the mechanism approximates the optimal solution of the resulting random instance and uses its value as a threshold to post linear prices to the remaining agents. We should highlight that this approach does not fit in the framework of online posted-price mechanisms, since it requires knowledge of the costs of a significant fraction of agents (which is obtained through bidding). 
\section{Model, Definitions and Preliminaries}
\label{sec:prelim}

%
A set function $f : 2^{\N} \to \reals_{\geq 0}$ is \emph{non-decreasing} (often referred to as \emph{monotone}) if for all $S \subseteq T \subseteq \N$, $f(S) \leq f(T)$. A function $f$ is \emph{submodular} if $f$ has non-increasing marginal values, i.e., for all $S \subseteq T \subseteq \N$ and all $i \not\in T$, $f(T \cup \{i\}) - f(T) \leq f(S \cup \{i\}) - f(S)$. We let $f_S(i) = f(S \cup \{ i \}) - f(S)$ denote the marginal value of $i$ with respect to $S$. In this work, we consider valuation functions $f : 2^{\N} \to \reals_{\geq 0}$ that are normalized, i.e. have $f(\emptyset) = 0$, monotone and submodular. 

A valuation function $f$ is accessed by \emph{value queries}, which for any given set $S \subseteq \N$, return $f(S)$. We let $\vmax = \max_{i \in \N} \{ f(\{ i \})\}$ denote the maximum value of any agent. In the following, all logarithms are base-$2$ unless stated otherwise. 

\vskip2pt\textit{Budget-Feasible Procurement Auctions and Mechanisms.}
We consider procurement auctions with a set $\N = \{1, \ldots, n\}$ of agents. Each agent $i$ has a private cost $c_i \in \reals_{\geq 0}$ for participating in the solution. The buyer has a budget $B \in \reals_{\geq 0}$ and aims to maximize a monotone submodular valuation $f : 2^{\N} \to \reals_{\geq 0}$, subject to the constraint that the sum of payments to the agents must be at most $B$. 

A (deterministic) \emph{direct revelation mechanism} $\mathcal{M} = (A, p)$ is a pair consisting of an allocation function and a payment function. The allocation function $A: \mathbb{R}_{\geq 0}^{n} \to 2^{\N}$ maps a bid vector $\vec{b} = (b_1, \ldots, b_n)$ submitted by the agents to a solution $A(\vec{b}) = S \subseteq \N$. The payment function $p: \mathbb{R}_{\geq 0}^{n} \to \mathbb{R}_{\geq 0}^{n}$ computes the payments $p(\vec{b})$ allocated to the agents for the solution $A(\vec{b})$. A randomized mechanism is a probability distribution over deterministic mechanisms. 

Given a mechanism $\mathcal{M} = (A, p)$, each agent $i$ aims to maximize their \emph{utility} $u_i(\vec{b})$ through their bid. It is $u_i(\vec{b}) = p_i(\vec{b}) - c_i$, if $i$ is included in the solution $A(\vec{b})$, and $u_i(\vec{b}) =  p_i(\vec{b})$, otherwise. 

A mechanism is \emph{budget-feasible} if for every bid vector $\vec{b}$, $\sum_{i \in \N} p_i(\vec{b}) \leq B$. A mechanism is \emph{individually rational}, if for every bid vector $\vec{b}$ and all agents $i \in \N$, $u_i(\vec{b}) \geq 0$. Hence,  $p_i(\vec{b}) \geq 0$ for all agents $i \in \N$, and $p_i(\vec{b}) \geq c_i$ for all agents $i$ included in the solution $A(\vec{b})$. Due to individual rationality and the budget constraint, we always let $p_i(\vec{b}) = 0$ for all agents $i \not\in A(\vec{b})$.

A mechanism is \emph{truthful}, if reporting their cost $c_i$ is a dominant strategy for every agent $i$, i.e., for all $i\in \N$, all $b \in \reals_{\geq 0}$ and every bidding vector $\vec{b}_{-i}$ of the other agents, $u_i(\vec{b}_{-i}, c_i) \geq u_i(\vec{b}_{-i}, b)$. 

For randomized mechanisms, that we consider in this work, we require that the mechanism is budget feasible and individually rational \emph{with certainty} and \emph{universally truthful}, i.e., the mechanism is a probability distribution over deterministic truthful mechanisms. 

\vskip2pt\textit{Online Budget-Feasible Procurement Auctions.}
In the online setting, agents arrive sequentially and the decision about whether the present agent $i$ is included in the solution and $i$'s payment is \emph{irrevocable} and is taken upon $i$'s arrival, using only information about the agents having arrived before $i$ and without any knowledge about future agent arrivals whatsoever. 

In this work, we consider an online procurement auctions with \emph{secretary arrivals}, where the set $\N$ of agents is chosen adversarially and the agents arrive in random order (formally, the agents' arrival order is drawn uniformly at random from all permutations of $\N$). 

\vskip2pt\textit{Posted-Price Mechanisms.}
We restrict our attention to (online) sequential \emph{posted-price} mechanisms. A posted-price mechanism decides on the price $p_i$ of each agent $i$ upon $i$'s arrival and makes $i$ a take-it-or-leave-it offer. Agent $i$ accepts or rejects the offer, depending on whether $p_i \geq c_i$ or not, without revealing any other information about their private cost $c_i$ whatsoever. If agent $i$ accepts the offer, $i$ is included in the mechanism's solution and the budget currently available decreases by $p_i$. Otherwise, we let $i$'s price be $p_i = 0$ and agent $i$ is discarded. 

Randomized posted-price mechanisms are individually rational and universally (and obviously) truthful, because among the two options available (accept or reject the mechanism's offer), the agents select the option that maximizes their utility $\max\{ 0, c_i - p_i \}$. Moreover, posted-price mechanisms are budget feasible, assuming that each price offered doed not exceed the budget currently available. 

We further restrict our attention to \emph{linear-price} mechanisms, which maintain a threshold $\hat{t}$ and offer a price $p_i = f_S(i) B / \hat{t}$ to every agent $i$ (assuming that $p_i$ does not exceed the budget currently available). The threshold $\hat{t}$ is an estimation of the optimal value and is maintained based on the marginal values of previous agents and their accept / reject decisions. 

We say that a linear-price mechanism is \emph{adaptive}, if it may use  a different threshold for determining the linear price offered to each agent, and \emph{non-adaptive}, if the mechanism decides on a fixed threshold before its first offer to an agent and uses it for all remaining agents.

\vskip2pt\textit{Competitive Ratio.}
We evaluate the performance of our online mechanisms using the competitive ratio \cite{BoroYan1998}. We compare the performance our our mechanism against the offline optimum, denoted $\opt$, which is the value of an optimal solution $S^\ast \subseteq \N$ that maximized $f(S^\ast)$ subject to $\sum_{i \in S^\ast} c_i \leq B$. We note that $\opt = f(S^\ast)$ is defined with respect to a computationally unrestricted algorithm with full advance knowledge of the set $\N$ of agents and their costs $c_1, \ldots, c_n$. 

In our setting, the competitive ratio $cr$ of an algorithm is the worst-case ratio, over all possible sets of agents $\N$, of $\opt$ divided by the algorithm's expected cost $\mathbb{E}[\alg]$ (where the expectation is taken over both the agents' arrival order and the algorithm's random choices).

\section{A Logarithmic Lower Bound for Non-Adaptive Linear-Price Mechanisms}
\label{sec:lower_bound}

To motivate the use of non-adaptive linear-price mechanisms, we first prove the following: 

\begin{theorem}\label{thm:non-adaptive-lower-bound} Every randomized non-adaptive linear-price mechanism is at least $\frac{\log n}{2}$-competitive.
\end{theorem}

\begin{proof}[Proofsketch.]
We consider a family of instances $\mathcal{I} = \{ I_0, I_1, ..., I_{\log n} \}$. Every instance consists of $n$ value-private cost pairs $(v,c)$, where for instance $I_i$, $v = 1$ and $c = B/2^i$, for every $i \in \{0,1,...,\log n\}$. 

We construct a probability distribution $\mathcal{D}_{\mathcal{I}}$ over the family $\mathcal{I}$ of instances. For every $i = 0, 1, 2, \ldots, \log(n)-1$, we let $I_{i}$ occur with probability $p_{i} = 1/2^{i+1}$ in $\mathcal{D}_{\mathcal{I}}$. We let $I_{\log n}$ occur with probability $p_{\log n} = 1/n$ in $\mathcal{D}_{\mathcal{I}}$, so that the sum of $p_i$s is equal to $1$. To conclude the proof, we show that the expected value of any deterministic non-adaptive linear-price mechanism on $\mathcal{D}_{\mathcal{I}}$ is at most $1$, while the expected value of the optimal solution on $\mathcal{D}_{\mathcal{I}}$ is $\frac{\log n}{2}+1$. Then, the theorem follows from Yao's principle 
(\cite[Chapter~8.4]{BoroYan1998} and \cite{Yao1977}). We provide a detailed proof in Section~\ref{sec:app:lower_bound}.
\end{proof}

The key idea behind Theorem~\ref{thm:non-adaptive-lower-bound} is that a non-adaptive linear-price mechanism has to decide on a fixed threshold $\hat{t}$ before getting any information about the agents' private costs (which for posted-price mechanisms is obtained only by offering prices to the agents and collecting their accept / reject responses). Hence, essentially the best approach of non-adaptive linear-price mechanisms is to learn $v_{\max}$ and select a threshold $\hat{t}$ from $\{ 1, 2, 4, \ldots, 2^{\log n}\}$ uniformly at random (which is what the posted-price mechanism of \cite[Section~4]{Bada2012} actually does). Hence, to improve on the logarithmic competitive ratio of \cite[Section~4]{Bada2012}, we next present an adaptive linear-price mechanism. 
\def\Brem{B_{\text{rem}}}
\def\Nrem{\mathcal{N}_{\text{rem}}}
\def\tinit{t_{\text{init}}}
\def\tmin{t_{\min}}
\def\tmax{t_{\max}}
\def\val{\text{val}}
\def\Event{\mathcal{E}}
\def\Bin{\mathcal{B}}


\begin{figure}[t]
\centering\includegraphics[width=0.6\textwidth]{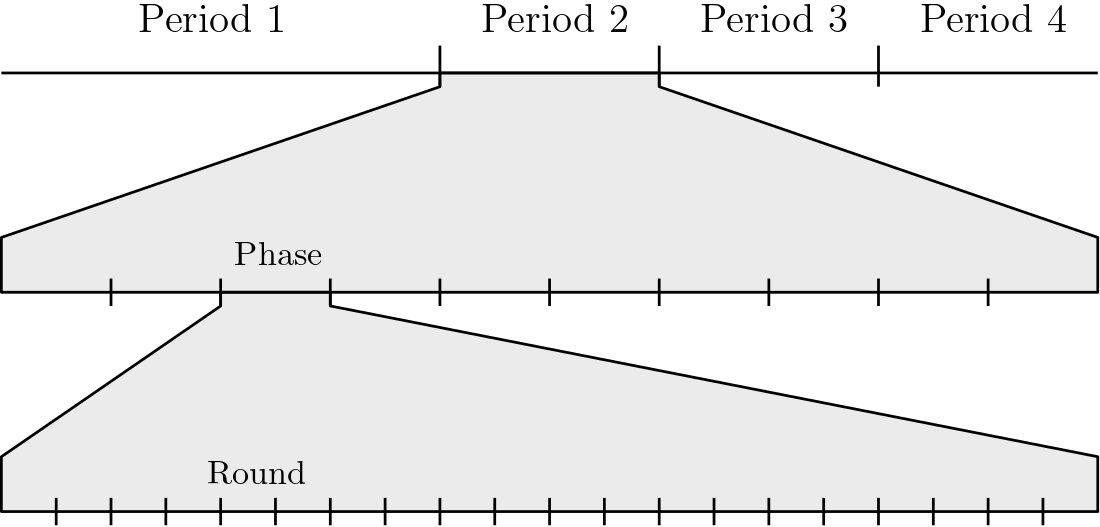}
    \caption{The input sequence is partitioned into four periods. Periods 2 to 4 are further partitioned into phases, with each phase consisting of multiple rounds.}
    \label{fig:phases}
\end{figure}

\begin{algorithm}[t]
\caption{LM-Mechanism}\label{alg:LMMECH}
\begin{algorithmic}[1]
\STATE \textbf{Input}: Set $\N$ of agents arriving in random order, budget $B$
\STATE \textbf{Initialization:} Solution $S \gets \emptyset$ and $\val \gets 0$, budget available $\Brem \gets B$, agents available $\Nrem \gets \N$
\STATE $\vmax, \Nrem \gets \text{LearningMaxValue}(\Nrem)$
\STATE $\tmin, \tmax, (S, \val, \Nrem, \Brem) \gets \text{PowerTowerSearch}((S, \val, \Nrem, \Brem), B, [\vmax, n\cdot\vmax])$
\STATE $\tinit, (S, \val, \Nrem, \Brem) \gets \text{BinarySearch}((S, \val, \Nrem, \Brem), B, [\tmin, \tmax])$
\STATE $(S, \val) \gets \text{Exploitation}((S, \val, \Nrem, \Brem), B, \tinit, [\tmin, \tmax])$
\RETURN $(S, \val)$
\end{algorithmic}
\end{algorithm}


\section{Posted-Price Mechanism: Building Blocks and Outline of the Analysis}
\label{sec:mech}

We next present the main steps of our posted-price \hyperref[alg:LMMECH]{LM-Mechanism} and the key ideas of its analysis. \hyperref[alg:LMMECH]{LM-Mechanism} assumes that $\opt / \vmax$ is sufficient large (LM stands for \emph{Large Market}). E.g., it suffices that $\opt / \vmax \geq 10^7$. In Section~\ref{sec:MainTheorem}, we use Dynkin's algorithm \cite{Dynkin63} and an adaptation of the posted-price mechanism in \cite[Section~4]{Bada2012} to deal with the case where $\opt / \vmax < 10^7$. 

As explained in Section~\ref{sec:approach}, \hyperref[alg:LMMECH]{LM-Mechanism} consists of $4$ periods. Each of the first three periods ``consumes'' a constant fraction of the agent sequence (and at most a constant fraction of the budget $B$) and provides the next period with a more refined estimation of $\opt$. More specifically, \hyperref[alg:LMMECH]{LM-Mechanism} proceeds along the following periods: 
\begin{description}
\item [Period 1 -- LearningMaxValue:] The first period calls Mechanism~\ref{alg:Vmax}, usually referred to as LearningMaxValue, which 
examines about $1/3$ of the agent sequence and returns the maximum value of an agent in this part. We let $\Event_1$ denote the event that the value returned by \hyperref[alg:Vmax]{LearningMaxValue} is indeed $\vmax$, which occurs with probability $1/3$. Assuming $\Event_1$, we have a first rough estimation of $\opt$, since $\vmax \leq \opt \leq n\cdot\vmax$. 

\item [Period 2 -- PowerTowerSearch:] The second period calls Mechanism~\ref{alg:PowerTower}, usually referred to as PowerTowerSearch,
which decides on two consecutive intervals $A = [a_{\min}, a_{\max}]$ and $B = [b_{\min}, b_{\max}]$ such that $a_{\max} = 2^{\frac{a_{\min}}{\vmax}}a_{\min}$ and $b_{\max} = 2^{\frac{b_{\min}}{\vmax}}b_{\min}$. PowerTowerSearch then selects one of the two intervals uniformly at random and returns its endpoints as $t_{\min}$ and $t_{\max}$. We let $\Event_2$ denote the event that $\opt \in A \cup B$ and $\Event_5$ denote the event that conditional on $\Event_2$, the right interval of $A$ and $B$ is chosen. 
In Section~\ref{sec:SecondPeriod}, we discuss the implementation of \hyperref[alg:PowerTower]{PowerTowerSearch} and prove that $\Event_2$ occurs with probability at least $0.9$. 

\item [Period 3 -- BinarySearch:] The third period calls Mechanism~\ref{alg:binarysearch}, usually referred to as BinarySearch, which performs binary search in the interval $[\tmin, \tmax]$ and returns $\tinit$. Assuming a proper execution of BinarySearch, which is implied by event $\Event_3$ formally defined later on, we have that $\opt / 4 \leq \tinit \leq O(\log\log(\tmax/\tmin))\opt$. The analysis of \hyperref[alg:binarysearch]{BinarySearch} is presented in Section~\ref{sec:ThirdPeriod}.

\item [Period 4 -- Exploitation:] The fourth period calls Mechanism~\ref{alg:exploitation}, usually referred to as Exploitation, which performs an adaptive search, starting with $\tinit$, on possible values of the mechanism's linear-price threshold that are powers of $2$ (see also \cite[Section~3]{Bada2012}). Assuming its proper execution, which is also implied by event $\Event_3$, Exploitation collects $\Omega(\opt)$ value from its part of the agent sequence. The analysis of \hyperref[alg:exploitation]{Exploitation} is presented in Section~\ref{sec:FourthPeriod}.
\end{description}

Different periods, as they run, update \hyperref[alg:LMMECH]{LM-Mechanism}'s state with consists of its current solution $S$, the total value collected so far $\val$, the budget available $\Brem$ and the sequence of agents not yet considered $\Nrem$. 
Each of the last three periods examine a (relatively small) constant fraction of the agent sequence. We let $\Event_4$ denote the event that the total number of agents examined by the $4$ periods of \hyperref[alg:LMMECH]{LM-Mechanism} is at most $n$ (which along with events $\Event_2$ and $\Event_3$ also imply that the total budget expended is at most $B$). The correctness of \hyperref[alg:LMMECH]{LM-Mechanism} is implied by event $\Event = \Event_1 \cap \Event_2 \cap \Event_3 \cap \Event_4 \cap \Event_5$. The correctness probability of \hyperref[alg:LMMECH]{LM-Mechanism} is analyzed in Section~\ref{sec:successProb}.


The two key technical ingredients in the analysis of \hyperref[alg:LMMECH]{LM-Mechanism} is a partitioning of the last $3$ periods into phases and rounds (as depicted Fig.~\ref{fig:phases}) and \hyperref[alg:TestTHRESHOLD]{TestEstimate}, which is applied to each round and tests if the current threshold is too large or too small. 
%
%
We next discuss their main properties. 


\subsection{Partitioning into Rounds and Phases}
\label{sec:rounds}

To partition the agent sequence into rounds, we employ a standard technique \cite{IndependentTrick}, where round lengths are drawn from a binomial distribution to ensure independence among agents and that the resulting prefix of the sequence is a \emph{random subset} of $\N$, where each agent participates with a given probability (often referred to as the round's \emph{participation probability} and denoted by $q$). 

Formally, for any suffix $\N' \subseteq \N$ of the agent sequence and any $a \in (0, 1)$, an \emph{$a$-round} (or simply a \emph{round}, if $a$ is clear from the context) is a prefix of $\N'$ with length drawn from a binomial distribution $\Bin(|\N'|, a)$ with parameters $|\N'|$ and $a$. We usually refer to $a$ as the the \emph{length parameter} (or the \emph{length probability}) of the round. 
Along with the random arrival order of the agents, selecting the length of a round as a binomially distributed random variable ensures that every two agents in $\N'$ are included in the given round independently with probability $a$ (which allows us to analyze the properties of each round using standard concentration inequalities). 

Given the agent sequence $\N$, we define a sequence of $\kappa \geq 1$ rounds with length parameters $a_1, \ldots, a_\kappa$ as follows: we select the length $x_1$ of the $1$st round from the binomial distribution $\Bin(n, a_1)$, let $n_1 = n-x_1$ be the number of agents remaining, select the length $x_2$ of the $2$nd round from $\Bin(n_1, a_2)$, $\ldots$, let $n_i= n_{i-1}-x_i$ be the number of agents remaining after the $i$-th round, select the length $x_{i+1}$ of the $(i+1)$-th round from $\Bin(n_i, a_{i+1})$, $\ldots$, select the length $x_{\kappa}$ of the $\kappa$-th round from $\Bin(n_{\kappa-1}, a_{\kappa})$. The following analyzes the distribution of the resulting random subsets of agents: 

\begin{lemma}\label{lemma:RoundPro}
For any $\kappa \geq 1$, we consider a sequence of round lengths distributed binomially as described above with length parameters $a_1, \ldots, a_\kappa$. Then, each agent $b \in \N$ is included in round $\kappa$ independently of the other agents with participation probability
\begin{equation}\label{eq:participation}
                q_{\kappa} = a_{\kappa}\prod_{i=1}^{\kappa-1}(1-a_{i})
\end{equation}
\end{lemma}

\begin{proof}[Proofsketch.]
The proof consists of two parts. First we observe that due to secretary agent arrivals, the probability that an agent belongs to any fixed interval of length $n_i$ of an agent sequence of length $n$ is $n_i / n$. Applying standard properties of the binomial distribution repeatedly, we derive the desired participation probability. Next, we use secretary agent arrivals and the binomial distribution of round lengths and show that the probability that any two agents belong to round $\kappa$ is $q_\kappa^2$. 
%
%
A detailed proof is given in Appendix~\ref{A1}. 
\end{proof}

In our mechanism, we draw multiple rounds with significantly varying lengths, which complicates the calculation of each round's participation probability using \eqref{eq:participation}. The following lower bound on the participation probability allows us to argue about the properties of each round using its length parameter $a$. The proof requires the precise definition of the length parameters of all rounds considered by our mechanism and is deferred to Appendix~\ref{sec:PART}. For the proof, we use that the length parameter of every round is $o(1)$ and that their sum over all rounds is significantly less than $1$ (which also allows us to show that event $\Event_4$ occurs with probability at least $0.97$, see Lemma~\ref{lemma:Fit} in Appendix~\ref{A3}).  

\begin{lemma}\label{lemma:PART}
The participation probability of any round with length parameter $a$ is at least $6a/(10e)$. 
\end{lemma}



\textit{Good, Dense Rounds and Phases.}
%
%
Our goal is to use TestThreshold on rounds in order to assess wether a threshold $t$ underestimates the true optimum. To this end, we will call an $\left(a,t\right)$-round, a round with length parameter $a$ on which we test a threshold $t$. For our test to be correct consistently we need the expected number of agents included in both the optimal solution and our $\left(a,t\right)$-round to be at least a fixed constant, i.e. $a\ge \frac{\vmax}{\opt}$. Since we don't know the value of $\opt$ and $t$ is our estimate of it, we say that an $\left(a,t\right)$-round is `good' if $a\ge 81\, e\, \vmax/t$. 

We allocate a budget $B_j = 3 \, C \cdot a_j \cdot B$ to every round $j$ with length parameter $a_j$, where $C = 1/(8e)$ is a fixed constant used in the following sections. For each round $j$ with set of agents $\N_j$ and length parameter $a_j$, we let $S^\ast_j \subseteq \N_j$ be the set of agents in round $j$ with maximum value $f(S^\ast_j)$ subject to the budget constraint $\sum_{ i \in S_j^\ast} c_i \leq B_j$, and let $\opt_j = f(S^\ast_j)$ be the optimal value of round $j$. 

We say that an $\left(a_j,t_j\right))$-round $j$ is \emph{dense} if $\opt_j \geq C \cdot a_j \cdot \opt$. Intuitively, a round is dense if it gets a fair share (with respect to its length parameter) of the optimal value. The following shows that `good' rounds are dense with good probability. 

\begin{lemma}\label{lemma:roundProb}
A  `good' $\left(a_j,t_j\right)$-round $j$, such that $t_j\le \opt$, is dense with probability at least $0.9$.
\end{lemma}

The high level intuition behind the proof of Lemma~\ref{lemma:roundProb} is that a good round is a random set that includes, in expectation, sufficiently many agents of the optimal solution. The requirement on the round being `good' is mostly technical and is used in the proof of Lemma~\ref{lemma:roundProb}. Lemma~\ref{lemma:roundProb} is key to the  analysis of \hyperref[alg:LMMECH]{LM-Mechanism} and can be found in Section~\ref{proof4.5}.

Applying \hyperref[alg:TestTHRESHOLD]{TestEstimate} to a dense round with threshold $t \leq \opt / 4$ should verify that $t$ is not too large and can be increased. To amplify the confidence of our test, we apply \hyperref[alg:TestTHRESHOLD]{TestEstimate} to a sequence of rounds. A \emph{$\left(\delta,a,t\right)$-phase} is a sequence of $\delta$ consecutive rounds all with the same length parameter $a$ and all testing the same threshold $t$ (see also Fig.~\ref{fig:phases}). 
We note that \hyperref[alg:TestTHRESHOLD]{TestEstimate} is always applied with the same threshold $t$ to every round of a given phase.
A $\left(\delta,a,t\right)$-phase is \emph{dense} if at least $\delta / 2$ of its rounds are dense. In Appendix~\ref{negdep}, we show that: 

\begin{lemma}\label{lemma:PhaseProbability}
A $\left(\delta,a,t\right)$-phase, consisting of `good' rounds and with $t\le \opt$, is dense with probability at least $1-\exp(\frac{-4\cdot\delta}{45})$.
\end{lemma}

For the proof of Lemma~\ref{lemma:PhaseProbability}, we observe that the events that different $a$-rounds are dense are negatively associated, in the sense that conditional on a round not being dense, the event that another round is dense becomes more likely. Then, Lemma~\ref{lemma:PhaseProbability} by an application of standard concentration bounds. 

\begin{remark}We should highlight that the randomness in the agent arrivals and the internal randomness of \hyperref[alg:LMMECH]{LM-Mechanism} are only used (i) to partition the agent sequence into rounds consisting of random subsets of agents; and (ii) to select the right interval between $A$ and $B$ in Period $2$. Everything else in \hyperref[alg:LMMECH]{LM-Mechanism} is deterministic. In fact, once the rounds are formed, \hyperref[alg:TestTHRESHOLD]{TestEstimate} and the other parts of  \hyperref[alg:LMMECH]{LM-Mechanism} operate without assuming anything about the order in which the agents arrive. \qed
 \end{remark}

\subsection{Estimating the Probability of Success}
\label{sec:successProb}

In the following, we always assume that \hyperref[alg:LMMECH]{LM-Mechanism} gathers positive value only in realizations where event $\Event$ holds (and always condition on $\Event$, even without mentioning that). As discussed after the outline of \hyperref[alg:LMMECH]{LM-Mechanism}, $\Event$ is the conjunction of events $\Event_1, \ldots, \Event_5$ formally defined below.

    \begin{enumerate}
        \item [$\Event_1$:] An agent with the maximum utility $\vmax$ is included in Period 1.
        
        \item [$\Event_2$:] In Period 2, the interval $[\vmax, n \cdot \vmax]$ is partitioned into sub-intervals. \hyperref[alg:PowerTower]{PowerTowerSearch} identifies two consecutive intervals $A$ and $B$. Then, $\Event_2$ denotes that $\opt / 8 \in A \cup B$.
        
        \item [$\Event_3$:] Every phase during periods $3$ and $4$ is dense with respect to the threshold $\hat{t}$ used by \hyperref[alg:TestTHRESHOLD]{TestEstimate} for that phase.

        \item [$\Event_4$:] The sum of the realized round lengths drawn by \hyperref[alg:LMMECH]{LM-Mechanism} is less than $n$.
        
        \item [$\Event_5$:] Given $\Event_2$, the interval $I$ of $A$ and $B$ chosen uniformly at random satisfies $\opt / 8 \in I$. 
    \end{enumerate}

The following shows that event $\Event$ occurs with constant probability/

\begin{lemma}\label{lemma:GoodEvent}
    Event $\Event$ happens with probability at least $1/20$.
\end{lemma}
\begin{proof}
In Appendix~\ref{A3}, we show that $\Prob{\Event_2} \ge 0.9$ in Lemma~\ref{lemma:E2}, $\Prob{\Event_3} \ge 0.9$ in Lemma~\ref{lemma:BinaryProb} and $\Prob{\Event_4} \geq 0.97$ in Lemma~\ref{lemma:Fit}.
We also have that $\Prob{\Event_1} \geq 1/3$ due to random agent arrivals and the properties of the binomial distribution. Finally, we argue that $\Prob{\Event_5\,|\,\Event_2} = 1/2$. Given that the set of possible realizations is restricted to a subset of those where event $\Event_2$ occurs, $\Event_5$ occurs with probability $1/2$, independently of any other events in the realizations considered in the conditioning on $\Event_2$. I.e., we can think that once we condition that the event $\Event_2$ occurs, the mechanism flips a fair coin (which is completely independent from $\Event_2$ and also from $\Event_1 \cap \Event_3 \cap \Event_4$): with probability $1/2$, the mechanism selects the right interval (and proceeds collecting value), and with probability $1/2$, the mechanism selects the wrong interval (and does not collect anything). Then, 
\[
\Prob{\lnot \Event_1 \cup \lnot \Event_2 \cup \lnot \Event_3 \cup \lnot \Event_4} \leq 2/3 + 0.1 +0.1 + 0.03 \leq 0.9 \]
Putting everything together:
$\Prob{\mathcal{E}} = \Prob{\Event_1 \cap \Event_2 \cap \Event_3 \cap \Event_4} \cdot \Prob{\Event_5 \,|\,\Event_2} \geq \frac{1}{10} \cdot \frac{1}{2} = 1/20$.
%
\end{proof}

\subsection{Estimation Testing}
\label{sec:estimate}

The key building block of our \hyperref[alg:LMMECH]{LM-Mechanism} is an online test, usually referred to as \hyperref[alg:TestTHRESHOLD]{TestThreshold}, which is applied to a $(\delta,a,t)$-phase. \hyperref[alg:TestTHRESHOLD]{TestThreshold}, whose pseudocode can be found in Mechanism~\ref{alg:TestTHRESHOLD}, determines whether:

\begin{quote}
    \textit{it is possible to collect value at least $C\cdot a\cdot \hat{t}$, where $C = 1/(7e)$ by applying linear-pricing with threshold $\hat{t}$ to a round with length parameter $a$.}
\end{quote}

\hyperref[alg:TestTHRESHOLD]{TestThreshold} determines the set $\N_{\text{round}}$ of agents of each round by drawing their sizes from a binomial distribution $\Bin(|\Nrem|, a)$. Then, it offers every agent $b \in \N_{\text{round}}$ the corresponding linear price $p_b = f_S(b) B /\hat{t}$ and collects value $f_S(b)$, if the price $p_b$ is accepted. We write that \hyperref[alg:TestTHRESHOLD]{TestThreshold} succeeds in a round, if it collects value at least  $C\cdot a\cdot \hat{t}$ from the present round, and fails, otherwise. We write that \hyperref[alg:TestTHRESHOLD]{TestThreshold} succeeds in a phase (or just succeeds) when at least half of the rounds are successful.

\hyperref[alg:TestTHRESHOLD]{TestThreshold} is the only point where \hyperref[alg:LMMECH]{LM-Mechanism} makes offers to the agents, collects value and expends budget. Periods $2$ and $3$ determine the length parameters and the number of phases to which \hyperref[alg:TestTHRESHOLD]{TestThreshold} is applied, and use their success / fail responses to guide our adaptive search for a good estimate of $\opt$. Then, Period $4$ applies \hyperref[alg:TestTHRESHOLD]{TestThreshold} to phases with appropriate length parameters in order to collect a total value of $\Omega(\opt)$. 
If during the execution of \hyperref[alg:LMMECH]{LM-Mechanism}, either the budget or the agent sequence is exhausted, or an agent with value larger than $\vmax$ is found, event $\Event$ is violated and \hyperref[alg:TestTHRESHOLD]{TestThreshold} (and \hyperref[alg:LMMECH]{LM-Mechanism}) abort (without  any value).

\def\valest{\mathrm{val}\mbox{-}\mathrm{test}}
\begin{algorithm}[t]
  \caption{TestThreshold}\label{alg:TestTHRESHOLD}
 \begin{algorithmic}[1]
     \STATE \textbf{Input:} Current state $(S, \val, \Nrem, \Brem)$, budget $B$, phase size $\delta$, length parameter $a$, threshold $\hat{t}$.
     \STATE \textbf{Initialization:} $C \leftarrow \frac{1}{7e}$, successes$\gets 0$
    \FOR{$j=1$ to $\delta$}
    \STATE Draw $\tau \sim \Bin(|\Nrem|,a)$, $\mathcal{N}_{\text{round}} \leftarrow \Nrem[1 : \tau]$
    \STATE $\Nrem \leftarrow \Nrem \setminus \mathcal{N}_{\text{round}},i \leftarrow 1, \valest \leftarrow 0$
    \STATE \textbf{if} $\Nrem = \emptyset$ \textbf{then} abort \textbf{end if}
     \WHILE{$i \leq |\mathcal{N}_{\text{round}}| $}
     \STATE \textbf{if} $f(\mathcal{N}_{\text{round}}[i]) > v_{max}$ \textbf{then} abort \textbf{end if}
    \STATE \textbf{else} 
    \begin{ALC@g}
    \STATE $p_i \leftarrow f_S(\mathcal{N}_{\text{round}}[i]) \cdot B\,/\, \hat{t} $\label{step:price}
    \IF{$ p_i \geq c_i$}
    \STATE $S \gets S\cup \{\mathcal{N}_{\text{round}}[i]\}$
    \STATE $\val \gets \val + f_S(\mathcal{N}_{\text{round}}[i])$\ \  and\ \ $\valest \leftarrow \valest + f_S(\mathcal{N}_{\text{round}}[i])$
    \STATE $\Brem \leftarrow \Brem - p_i$ 
    \STATE \textbf{if} $\Brem \leq 0$ \textbf{then} abort \textbf{end if}
    \ENDIF
    \end{ALC@g}
    \STATE \textbf{end if}
    \STATE \textbf{if} $\valest \ge C \cdot a\cdot \hat{t}$ \textbf{then} $\text{successes} \gets \text{successes} +1$ \textbf{end if}\label{step:test}
    \STATE $i \gets i + 1$
    \ENDWHILE
    \ENDFOR
    \STATE \textbf{if} $\text{successes} \ge \delta /2$ \textbf{then} $(1, S, \val, \Nrem, \Brem)$ \textbf{else }$(0, S, \val, \Nrem, \Brem)$
 \end{algorithmic}
\end{algorithm}


We next consider applications of \hyperref[alg:TestTHRESHOLD]{TestThreshold} to `good' rounds/phases and establish its main properties. We highlight that every time \hyperref[alg:TestTHRESHOLD]{TestThreshold} is applied in periods $2$, $3$ and $4$, the length parameter $a$ of each round respects the goodness property with respect to the chosen $t$ (see Proposition~\ref{prop:good1} in Section~\ref{sec:SecondPeriod}, Proposition~\ref{prop:good2} in Section~\ref{sec:ThirdPeriod} and Proposition~\ref{prop:good3} in Section~\ref{sec:FourthPeriod}). 
In fact, this is the key objective behind how the four periods of \hyperref[alg:LMMECH]{LM-Mechanism} are structured: 
\hyperref[alg:PowerTower]{PowerTowerSearch}, \hyperref[alg:binarysearch]{BinarySearch} and 
\hyperref[alg:exploitation]{Exploitation} are carefully designed in order to ensure that  \hyperref[alg:TestTHRESHOLD]{TestThreshold} is applied to `good' rounds.
Therefore, due to Lemma~\ref{lemma:roundProb} (resp. Lemma~\ref{lemma:PhaseProbability}), every round (resp. phase) to which \hyperref[alg:TestTHRESHOLD]{TestThreshold} is applied is dense with probability at least $0.9$ (resp. close enough to $1$).

Due to the linear pricing scheme in step~\ref{step:price} and the success criterion in  step~\ref{step:test}, every application of \hyperref[alg:TestTHRESHOLD]{TestThreshold} to a `good' $(a,\hat{t})$-round (regardless of its success or failure) expends a budget of at most 
\[ \frac{C \cdot a \cdot \hat{t} + \vmax}{\hat{t}} \cdot B \leq \left( C \cdot a + \frac{\vmax}{\hat{t}}\right) B \leq 2 \cdot C \cdot a \cdot B\,,
\]
where the last inequality holds because the round is good. Thus:

\begin{property}\label{property:a}
When \hyperref[alg:TestTHRESHOLD]{TestThreshold} is applied to a `good' $\left(a,\hat{t}\right)$-round the budget expended does not exceed the round's budget $B_j = 3\cdot C\cdot a\cdot B$. 
\end{property}

Using Property~\ref{property:a}, that \hyperref[alg:TestTHRESHOLD]{TestThreshold} and that the sum of the length parameters of all rounds is less than $1$, we prove the following in Appendix~\ref{budgetproofs}:

\begin{lemma}\label{BUDGET}
The total budget expended by \hyperref[alg:LMMECH]{LM-Mechanism} is at most $B/10$.
\end{lemma}


\begin{property}\label{property:b}
\hyperref[alg:TestTHRESHOLD]{TestThreshold} succeeds (i.e., it collects value at least $C\cdot a \cdot \hat{t}$) when applied to an $\left(a_j,\hat{t}_j\right)$-round that is dense and $\hat{t}_j \le \opt / 4$.
\end{property}
\begin{proof}
Let $S_j$ be the part of the solution obtained by applying \hyperref[alg:TestTHRESHOLD]{TestThreshold} with threshold $\hat{t}$ to an $a_j$-round $j$ that is dense with respect to $\hat{t}$. We assume that the value of each agent $b$ is $f(b) \leq\vmax$, as otherwise \hyperref[alg:TestTHRESHOLD]{TestThreshold} aborts. Let also $S^\ast_j$ denote the optimal solution for round $j$.

Suppose some agents $ b \in S_j^{*} \setminus S_j $ were included in the optimal solution but not in $S_j$. This can happen for one of two reasons: either their private costs exceeded the mechanism's offer, i.e., $ c_b > f_{S_j}(b)\,B\,/\,\hat{t}$, or $\Brem$ was insufficient to make an appropriate offer to $ b $. However, by design, only the first case can occur. Due to our linear pricing scheme, no agent is ever offered a price larger than $\vmax\,B\,/\,\hat{t} \le \frac{B}{1024}$, where the inequality holds because we always use thresholds $\hat{t} \geq 1024\,\vmax$. Combining this with Lemma~\ref{BUDGET}, we conclude that all agents $ b \in S_j^{*} \setminus S_j $ declined our offer.

We can rewrite the sum of the values of these agents as:
\begin{equation}\label{eq:basic}
            \sum_{b\in S_j^{*} \setminus S_j} f_{S_j}(b) = \sum_{b\in S_j^{*} \setminus S_j} \frac{f_{S_j}(b)}{c_b} \cdot c_b < \sum_{b\in S_j^{*} \setminus S_j}\frac{\hat{t}}{B}\cdot c_b \le \hat{t}\cdot \frac{B_j}{B} = 3 \cdot C\cdot a_j \cdot \hat{t}
\end{equation}
The second to last inequality follows from the hypothesis $ c_b > f_{S_j}(b)\,B\,/\,\hat{t}$ and the last inequality follows from the fact that $S_j^{*}$ is budget-feasible with respect to $B_j = 3 \cdot C\cdot a_j\cdot B$.

Using monotonicity and submodularity of the valuation function we have:
\begin{equation}\label{eq:basic2}
    f(S_j^{*})-f(S_j) \le f(S_j^{*}\cup S_j) -f(S_j) \le \sum_{b\in S_j^{*}\setminus S_j} f_{S_j}(b)
\end{equation}

Applying \eqref{eq:basic} to \eqref{eq:basic2}, we get that $f(S_j) \geq f(S_j^{*}) - 3 \cdot C\cdot a_j \cdot \hat{t}$. Finally, using the hypothesis that round $j$ is dense, and hence $f(S^\ast_j) \geq C \cdot a_j \cdot \opt$, we obtain that:
\[
        f(S_j) \ge C \cdot a_j \cdot \opt -  3\cdot C\cdot a_j \cdot \hat{t}
\]
Since $\hat{t} \le \frac{\opt}{4}$, 
$f(S_j)\ge 4\cdot C \cdot a_j \cdot \hat{t} -3 \cdot C\cdot a_j\cdot\hat{t} = C\cdot a_j \cdot \hat{t}$.
Thus, \hyperref[alg:TestTHRESHOLD]{TestThreshold} it collects value at least $C\cdot a_j \cdot \hat{t}$ and succeeds. 
\end{proof}
 The Lemma above extends to phases the following way:
 \begin{lemma}\label{lemma:phasesucc}
     \hyperref[alg:TestTHRESHOLD]{TestThreshold} succeeds (i.e., it collects value at least $C\cdot a \cdot \hat{t}$ and returns $1$) when applied to a $\left(\delta_j,a_j,\hat{t}_j\right)$-phase that is dense and $\hat{t}_j \le \opt / 4$.
 \end{lemma}

The next four sections discuss how periods $1$ to $4$ are structured so that \hyperref[alg:TestTHRESHOLD]{TestThreshold} is applied to phases so that (conditional on event $\Event$) a total value of $\Omega(\opt)$ is collected. 




\section{First Period -- Learning the Maximum Value}
\label{sec:FirstPeriod}

To learn $\vmax$, we sample the first $\tau$ agents, where $\tau \sim \Bin(n,1/3)$. Conditional on event $\Event$, the agent with the maximum value $\vmax$ is among the first $\tau$ agents. In this case, we can use the highest observed value to construct the interval $[\vmax, n \cdot \vmax]$, which is guaranteed to contain $\opt$, due to the submodularity of $f$. The pseudocode for the first period is given in Mechanism~\ref{alg:Vmax}.

\begin{algorithm}[t]
  \caption{LearningMaxValue}\label{alg:Vmax}
 \begin{algorithmic}[1]
     \STATE \textbf{Input:} Current set $\Nrem$ of agents available 
     \STATE Sample $\tau \sim \Bin(|\Nrem|, 1/3)$, 
     $\mathcal{N}_{\text{period}} \leftarrow \Nrem[1 : \tau]$
    \STATE Offer price $p=0$ to the agents in $\mathcal{N}_{\text{period}}$
    \STATE Let $\vmax \leftarrow \max_{b \in \mathcal{N}_{\text{period}}}f(b)$ and $\Nrem \leftarrow \Nrem \setminus \mathcal{N}_{\text{period}}$
    \STATE \textbf{return} $\vmax, \Nrem$
 \end{algorithmic}
\end{algorithm}

\section{Second Period -- Power Tower Search}
\label{sec:SecondPeriod}

Next, we further refine our estimation of $\opt$, narrowing it down to two intervals of power tower length. To this end, we partition the interval $[\vmax,n\cdot \vmax]$ into an interval sequence  $\mathcal{T}$ as follows:
\begin{equation}
    \mathcal{T} = \left\{[t_1,t_2],(t_2,t_3],\dots,(t_T,t_{T+1}]\right\},
\end{equation} 
where $t_1 = \vmax, t_2 = 10^7\,\vmax$, $t_i = 2^{\frac{t_{i-1}}{\vmax}}\cdot t_{i-1}$ for $2<i<T+1$, and $t_{T+1}=n\cdot \vmax$, with $T=|\mathcal{T}|$. 
Using \hyperref[alg:TestTHRESHOLD]{TestThreshold}, we test all intervals $(t_i, t_{i+1}]$, for $i = 2, \ldots, T-1$, so that we find the correct interval to pick. We exclude $[t_1, t_2]$ and $(t_T, t_{T+1}]$, because the correct interval can be inferred without directly testing these values. To test each interval $(t_i, t_{i+1}]$, we apply \hyperref[alg:TestTHRESHOLD]{TestThreshold} with threshold $\hat{t} = t_i$ to a $(\gamma_i, a_i,t_i)$-phase, where the length parameter is $a_i = \frac{81\cdot e\cdot \vmax}{t_i}$ and the number of rounds is $\gamma_i =\frac{3}{2}\cdot\log\left(\frac{t_i}{\vmax}\right)$. The pseudocode of \hyperref[alg:PowerTower]{PowerTowerSearch} is given in Mechanism~\ref{alg:PowerTower}.

\begin{algorithm}[t]
  \caption{PowerTowerSearch}\label{alg:PowerTower}
 \begin{algorithmic}[1]
     \STATE \textbf{Input:} Current state $(S, \val, \Nrem, \Brem)$, budget $B$, interval $[\vmax, n \cdot \vmax]$.
    \STATE \textbf{Construct:} sequence $\mathcal{T}$ of interval s.t. $t_1 \gets \vmax, t_2 \gets 10^7\cdot \vmax$, $t_i \gets 2^{\frac{t_{i-1}}{\vmax}}\cdot t_{i-1}$ and $t_{T+1} \gets n\cdot \vmax$.
    \STATE \textbf{Initialization:} Length parameters $a_i = \frac{81\cdot e\cdot \vmax}{t_i}$ and phase sizes $\gamma_i =\frac{3}{2}\cdot\log\left(\frac{t_i}{\vmax}\right) $, $\text{index} \gets 1$.
    \FOR{$i=2$ to $|\mathcal{T}|$}
    \STATE $(\text{hit}, S, \val, \Nrem, \Brem) \leftarrow \text{TestThreshold}((S, \val, \Nrem, \Brem), B,\gamma_i, a_i, t_i )$
    \STATE \textbf{if} {$\text{hit} = 1$} \textbf{then} $\text{index} \gets i$ \textbf{end if}
    \ENDFOR
    \STATE With probability $1/2$: $t_{\min} \gets \mathcal{T}[\text{index}-1]$ and $t_{\max} \gets \mathcal{T}[\text{index}]$.
    \STATE Otherwise $t_{\min} \gets \mathcal{T}[\text{index}]$ and $t_{\max} \gets \mathcal{T}[\text{index}+1]$.
    \STATE \textbf{return} $t_{\min}, t_{\max}, (S, \val, \Nrem, \Brem)$.
 \end{algorithmic}
\end{algorithm}

A $(a_i, \gamma_i,t_i)$-phase is considered \emph{successful}, if \hyperref[alg:TestTHRESHOLD]{TestThreshold} applied with threshold $\hat{t} = t_i$ returns success in at least $\gamma_i/2$ rounds. Each length parameter $a_i$ is chosen such that every $(a_i,t_i)$-round is `good'.
 
\begin{proposition}\label{prop:good1}
    During PowerTowerSearch, all phases consist of `good' $(a_i,t_i)$-rounds.
\end{proposition}

The phase lengths $\gamma_i$ are carefully chosen so that (i) the total number of agents examined during this period is at most a constant fraction of $n$; and (ii) we ensure that a phase with threshold much larger than $\opt$ cannot be successful. Formally:

\begin{lemma}\label{lemma:PERIOD1}
    Let $\frac{\opt}{8}\in (t_{i-1}, t_{i}]$. Then,   for every $j \geq i+1$, the $(\gamma_j, a_j,t_j)$-phase tested with threshold $\hat{t} = t_j$ is not successful. 
\end{lemma}

\begin{proof}
    Without loss of generality, let $t_{i} = \frac{\opt}{8}$. We prove that the $(\gamma_{i+1}, a_{i+1},t_{i+1})$-phase tested with threshold $\hat{t} = t_{i+1}$ cannot be successful (if proven for $j=i+1$, this must also apply to every $j > i+1$). 
    
    If the $(\gamma_{i+1}, a_{i+1},t_{i+1})$-phase tested with threshold $\hat{t} = t_{i+1}$ was successful, the total value that we would have collected during that phase would be:
\[
            \text{Value} \ge \frac{\gamma_{i+1}}{2} \cdot C \cdot a_{i+1} \cdot t_{i+1}  
\]
    Using the definitions of $a_i$ and $\gamma_i$, we get that:
\[
         \text{Value}\ge \frac{\frac{3}{2}\cdot \log\left(\frac{t_{i+1}}{\vmax}\right)}{2} \cdot C \cdot \frac{81\cdot e \cdot \vmax}{t_{i+1}} \cdot t_{i+1} 
\]
    Since $t_{i+1} = 2^{t_{i}/\vmax}\cdot t_i$ and using that $C = 1/(7e)$, we obtain that:
    \begin{equation}
         \text{Value}> 8\cdot\log\left(2^{\frac{t_i}{\vmax}}\cdot \frac{t_{i}}{\vmax}\right)\cdot \vmax
    \end{equation}
    Since $2^{\frac{t_i}{\vmax}}\cdot \frac{t_{i}}{\vmax} \ge 2^{\frac{t_i}{\vmax}}$, we get 
    $\text{Value} > 8 \cdot t_i =\opt$, a contradiction to the definition of $\opt$.
\end{proof}

After testing all $(\gamma_i, a_i,t_{i})$-phases corresponding to intervals $(t_i, t_{i+1}]$, each with threshold $t_i$, we select the interval corresponding to the last successful phase and the interval preceding it. I.e., if $(t_i, t_{i+1}]$ corresponds to the last successful phase, we select the intervals $(t_{i-1}, t_i]$ and $(t_i, t_{i+1}]$. In the definition of event $\Event_2$, we require that either $\frac{\opt}{8} \in (t_{i-1}, t_{i}]$ or $\frac{\opt}{8} \in (t_{i}, t_{i+1}]$, where $(t_{i-1}, t_i]$ and $(t_i, t_{i+1}]$ are the intervals selected by PowerTowerSearch. Let $i^\ast$ be such that $\frac{\opt}{8} \in (t_{i^{*}}, t_{i^{*}+1}]$. Then, for the event $\Event_2$ to occur, we require that either the $(\gamma_{i^\ast}, a_{i^\ast}, t_{i^\ast})$-phase or the $(\gamma_{i^\ast+1}, a_{i^\ast+1}, t_{i^\ast+1})$-phase is the last successful phase. 

We next show that the event $\Event_2$ is implied by the event that the $(\gamma_{i^\ast}, a_{i^\ast}, t_{i^\ast})$-phase is dense. So, let us assume that the $(\gamma_{i^\ast}, a_{i^\ast}, t_{i^\ast})$-phase is dense, which occurs with probability at least $0.9$, by Lemma~\ref{lemma:E2} (which, in turn, follows from Lemma~\ref{lemma:PhaseProbability}). Hence, with probability at least $0.9$, the $(\gamma_{i^\ast}, a_{i^\ast}, t_{i^\ast})$-phase is successful. If the $(\gamma_{i^\ast}, a_{i^\ast}, t_{i^\ast})$-phase is the last successful phase, then we select the intervals $(t_{i^\ast-1}, t_{i^\ast}]$ and $(t_{i^\ast}, t_{i\ast+1}]$ and the event $\Event_2$ occurs. Otherwise, Lemma \ref{lemma:PERIOD1} implies that the $(\gamma_{i^\ast+1}, a_{i^\ast+1}, t_{i^\ast+1})$-phase must be the last successful phase. Then, we select the intervals $(t_{i^\ast}, t_{i^\ast+1}]$ and $(t_{i^\ast+1}, t_{i\ast+2}]$ and the event $\Event_2$ occurs again.

\section{Third Period -- Binary Search}
\label{sec:ThirdPeriod}
  
  We next show how to apply \emph{binary search} to the interval returned by Period $2$. To apply binary search to the given interval $(t_{\min},t_{\max}]$, it is essential that the feedback on each estimation is correct with high probability. To achieve this, we test each estimation $\hat{t}$ of the optimum on $(m,a,\hat{t})$-phases, where $m=\left\lceil 8\cdot \log\log\left(\frac{t_{\max}}{t_{\min}}\right)\right\rceil $ rounds and $a=\frac{1}{6\,\left\lceil\log \log\left(\frac{t_{\max}}{t_{\min}}\right)\right\rceil \cdot m}$. The value of $m$ is set so that we can lower bound the probability of the event $\Event_3$ by union bound on the number of different phases. The length parameter $a$ is set so that the total number of phases used in Period $3$ times the number $m$ of rounds in each phase times $a$ is at most $1/3$ (which implies that Period $3$ ``consumes'' about $1/3$ of the agent sequence). We next show that this choice of $a$ makes our rounds good:
  \begin{proposition}\label{prop:good2}
    $(a,t_{\min})$-rounds (and thus all rounds) during BinarySearch are good.
  \end{proposition}
  
  \begin{proof}
  We need to verify that $\frac{t_{\min}}{\vmax} \ge \frac{81\,e}{a}$. Using our choice of $a$, we obtain that: 
      \begin{equation}
          \begin{split}
      \frac{t_{\min}}{\vmax} \ge 486\cdot e \cdot \left\lceil\log\log\left(\frac{t_{\max}}{t_{\min}}\right)\right\rceil\cdot \left\lceil 8\cdot  \log\log\left(\frac{t_{\max}}{t_{\min}}\right)\right\rceil 
          \end{split}
      \end{equation}
      Using the fact that $t_{\max} = 2^{\frac{t_{\min}}{\vmax}}\cdot t_{\min}$, we get:
      \begin{equation}
          \frac{t_{\min}}{\vmax}\ge 486\cdot e\cdot \left\lceil\log\left(\frac{t_{\min}}{\vmax}\right)\right\rceil\cdot \left\lceil 8\cdot  \log\left(\frac{t_{\min}}{\vmax}\right)\right\rceil\,, 
      \end{equation}
      which is true for $\frac{t_{\min}}{\vmax}\ge  10^7$.
  \end{proof}
The number of rounds needed to conduct BinarySearch is $\ell = \frac{1}{6\cdot a}$. The pseudocode for BinarySearch is presented in Mechanism~\ref{alg:binarysearch}.

\begin{algorithm}[t]
\caption{BinarySearch}\label{alg:binarysearch}
\begin{algorithmic}[1]
\STATE \textbf{Input}: Current state $(S, \val, \Nrem, \Brem)$, budget $B$, search interval $[t_{\min},t_{\max}]$
\STATE \textbf{Initialization:} Phase size $m \gets \left \lceil 8\cdot\log\log\left(\frac{t_{\max}}{t_{\min}}\right)\right\rceil$\,, length parameter $a \gets \frac{1}{6\,\left\lceil\log\log\left(\frac{t_{\max}}{t_{\min}}\right)\right\rceil\cdot m}$, \\
\hspace*{2.1cm}$\text{low} \gets 0$, 
$\text{high} \gets \left\lceil \log\left(\frac{t_{\max}}{t_{\min}}\right) \right\rceil$, $\text{mid} \gets \left\lceil (\text{high} + \text{low})/2 \right\rceil$
    \WHILE{$\text{low} \leq \text{high}$}
    \STATE $\text{hit} \gets 0$
    \STATE $(\text{hit}, S, \val, \Nrem, \Brem) \leftarrow \text{TestThreshold}((S, \val, \Nrem, \Brem), B,m, a, 2^{\text{mid}}\cdot t_{\min})$
    \STATE \textbf{if} $\text{hit} = 1$ \textbf{then} $\text{low} \gets \text{mid}$, $\text{mid} \gets \lceil (\text{high} + \text{low})/2 \rceil$.
    \STATE \textbf{else} $\text{high} \gets \text{mid}$, $\text{mid} \gets \lfloor \text{high} + \text{low})/2 \rfloor$ \textbf{end if}
    \ENDWHILE
    \STATE \textbf{return} $(2^{\text{mid}}\cdot t_{\min}, (S, \val, \Nrem, \Brem))$.
\end{algorithmic}
\end{algorithm}

By Lemma~\ref{lemma:phasesucc}, if all thresholds not greater than $\frac{\opt}{4}$ are tested on dense phases, we cannot end up with a substantial underestimation of $\opt$ after conducting BinarySearch. Below we prove an even sharper bound on the possible estimates that BinarySearch may return.
\begin{lemma}\label{lemma:dist}
    Let $\hat{t}$ be the estimate \hyperref[alg:binarysearch]{BinarySearch} returns, under event $\mathcal{E}$. Then:
    \begin{equation}
        \frac{\opt}{8} \le \hat{t} \le 84\cdot e\cdot \left\lceil\log\log\left(\frac{t_{\max}}{t_{\min}}\right)\right\rceil\cdot \opt
    \end{equation}
\end{lemma}
\begin{proof}
    Suppose $\hat{t}<\frac{\opt}{8}$. This would imply that an intermediate estimate $\hat{t}' \in \left[\frac{\opt}{8},\frac{\opt}{4}\right]$ failed in a dense phase. However, this cannot happen under event $\mathcal{E}$, due to Lemma  \ref{lemma:phasesucc}.
    Now consider the case where $\hat{t}> 84 \, e\, \left\lceil\log\log\left(\frac{t_{\max}}{t_{\min}}\right)\right\rceil\cdot \opt$. This would mean that a dense phase succeeded using an overestimated value  $\hat{t}' >84\, e\, \left\lceil\log\log\left(\frac{t_{\max}}{t_{\min}}\right)\right\rceil\cdot \opt$. Specifically, in that phase, the accumulated value would be:
    \begin{equation}
        \begin{split}
            \text{Value}\ge \frac{m}{2}\cdot C \cdot a \cdot \hat{t}' 
        \end{split}
    \end{equation}
    By using the assumption of the overestimation $\hat{t}'$, along with the definitions of $C$ and $a$ we get:
    \begin{equation}
        \begin{split}
            \text{Value}> \frac{m}{2}\cdot \frac{1}{7e} \cdot \frac{1}{6\left\lceil\log \log\left(\frac{t_{\max}}{t_{\min}}\right)\right\rceil \cdot m} \cdot 84\, e\, \left\lceil\log\log\left(\frac{t_{\max}}{t_{\min}}\right)\right\rceil\cdot \opt=\opt
        \end{split}
    \end{equation}    
    which clearly contradicts the definition of $\opt$.

   Thus, we conclude that $\hat{t} \in \left[\frac{\opt}{8},\ \ 84 \, e\,\log\log\left(\frac{t_{\max}}{t_{\min}}\right)\cdot \opt\right]$
\end{proof}

\section{Fourth Period -- Exploitation}\label{sec:FourthPeriod}

Lemma \ref{lemma:dist} indicates that, after running BinarySearch, our estimation is at most $84\cdot e \cdot \left\lceil\log\log\left(\frac{t_{\max}}{t_{\min}}\right)\right\rceil$ times $\opt$. In the final stage of our mechanism, we aim to converge to a threshold $\hat{t}$ close to $\opt$ and apply $\hat{t}$ to a constant fraction of agents, thereby collecting a total value of $\Omega(\opt)$. 

\begin{algorithm}[t]
\caption{Exploitation}\label{alg:exploitation}
\begin{algorithmic}[1]
\STATE \textbf{Input}: Current state $(S, \val, \Nrem, \Brem)$, budget $B$, initial threshold $\tinit$, search interval $[t_{\min}, t_{\max}]$
\STATE \textbf{Initialization:} $\hat{t} \gets \tinit$, Phase size $m \gets \left \lceil 8\cdot\log\log\left(\frac{t_{\max}}{t_{\min}}\right)\right\rceil$\,, length parameter $a \gets \frac{1}{6\,\left\lceil\log\log\left(\frac{t_{\max}}{t_{\min}}\right)\right\rceil\cdot m}$
    \FOR{$i = 1$ to $\frac{1}{6 \cdot m \cdot a} = \left\lceil\log\log\left(\frac{t_{\max}}{t_{\min}}\right)\right\rceil$}
    \STATE $(\text{hit}, S, \val, \Nrem, \Brem) \leftarrow \text{TestThreshold}((S, \val, \Nrem, \Brem), B, a, \hat{t})$
    \STATE \textbf{if} $\text{hit} =1$ \textbf{then} $\hat{t} \gets 2 \cdot \hat{t}$ \textbf{else} $\hat{t} \gets \hat{t}/2$ \textbf{end if}
    \ENDFOR
    \STATE \textbf{return} $(S, \val)$.
\end{algorithmic}
\end{algorithm}

\hyperref[alg:exploitation]{Exploitation}, whose pseudocode is presented in Mechanism~\ref{alg:exploitation}, is applied across  a total number of $\ell/m$ $(m,a,\hat{t})$-phases, thus using $\ell = 1/(6a)$ rounds in total. Under the event $\Event$, it holds that $\tmin \le \frac{\opt}{8}$. Thus, we can leverage Proposition~\ref{prop:good2} to claim the following:
\begin{proposition}\label{prop:good3}
    Conditioning on the event $\Event$, all rounds during Exploitation are good. 
\end{proposition}
The mechanism operates by doubling the estimation after every successful phase and halving it after every failed phase. Under the event $\Event$, the threshold $\hat{t}$ never falls below $\opt / 8$, by Lemma \ref{lemma:phasesucc}. Combining this with Lemma~\ref{lemma:dist}, we conclude that we can have no more than \[ \frac{\frac{\ell}{m} + \log\log\log\left(\frac{t_{\max}}{t_{\min}} \right)+ \log(84e)+2}{2} \] failed phases; exceeding this would cause our estimation to drop below \( \frac{\opt}{8} \). Consequently, we are guaranteed at least \( \frac{\frac{\ell}{m} -\log\log\log\left(\frac{t_{\max}}{t_{\min}} \right)- \log(84e)-2}{2}\) successful phases, which ensures a constant competitive ratio. Formally,

\begin{lemma} \label{lemma:compratio}
    Conditional on the event $\Event$, Exploitation collects a total value of at least $\frac{\opt}{4032e}$.
\end{lemma}
\begin{proof}
    As we have argued already, we are guaranteed to have at least $\frac{\frac{\ell}{m} -\log\log\log\left(\frac{t_{\max}}{t_{\min}} \right)- \log(84e)-2}{2}$ successful Phases, which in turn means that the value gathered is greater than:
\[
        \text{Value}\ge \underbrace{\frac{\frac{\ell}{m} -\log\log\log\left(\frac{t_{\max}}{t_{\min}} \right)- \log(84e)-2}{2}}_{\text{successful Phases}} \cdot \underbrace{\frac{m}{2} \cdot C \cdot a\cdot \frac{\opt}{8}}_{\text{Value of each Phase}} 
\]
    Now using the definitions of $a$, $\ell$, $C$, and $t_{\max}$, we get that
\[
        \text{Value} \ge\left(\frac{1}{6} - \frac{\log\log\log\left(2^{\frac{t_{\min}}{\vmax}}\right) + \log(84e)+2}{6\cdot\log\log\left(2^{\frac{t_{\min}}{\vmax}}\right)}\right)\cdot \frac{\opt}{224e}
\]
    Finally, by the monotonocity of $g(x) = \frac{-\log\log\log(x) - \log(84e)-2}{6\cdot\log\log(x)}$, along with the fact that $\frac{t_{\min}}{\vmax} \ge  10^7,$ we conclude that $\text{Value}  \ge \frac{1}{4032e}\cdot \opt$.
\end{proof}

\section{Putting Everything Together and Removing the Large Market Assumption}\label{sec:MainTheorem}




Below we present the proof of the competitiveness of \hyperref[alg:LMMECH]{LM-Mechanism}.

\begin{theorem}\label{theorem:lm}
    Assuming that $\opt > 10^{7}\cdot \vmax$, \hyperref[alg:LMMECH]{LM-Mechanism} is $O(1)$-competitive.
\end{theorem}

\begin{proof}
In the four periods of \hyperref[alg:LMMECH]{LM-Mechanism}, we assume that the event $\Event$ occurs. Lemma~\ref{lemma:GoodEvent} proves that $\Event$ happens with probability at least $1/20$. Conditional on the event $\Event$, the \hyperref[alg:LMMECH]{LM-Mechanism} achieves a competitive ratio of $4032e$ due to Lemma~\ref{lemma:compratio}. Overall, the expected total value of \hyperref[alg:LMMECH]{LM-Mechanism} is  at least $\frac{1}{20} \cdot \frac{1}{4032e}\cdot \opt$.
\end{proof}

Finally, we need to remove the Large Market Assumption that $\opt > 10^{7}\cdot \vmax$. To this end, we present the mechanism PostedPrices, whose pseudocode can be found in Mechanism~\ref{alg:PPMECH}.
\hyperref[alg:PPMECH]{PostedPrices} invokes three different mechanisms, with a constant probability each, which deal with instances of different market size (i.e., magnitude of $\frac{\opt}{\vmax}$). We are now ready to prove our main result:

\begin{algorithm}[t]
\caption{PostedPrices}\label{alg:PPMECH}
\begin{algorithmic}[1]
\STATE \textbf{Input}: Set of agents $\mathcal{N}$, budget $B$
\STATE With probability $0.1$, execute Dynkin's algorithm on $\N$.
\STATE With probability $0.1$, execute \hyperref[alg:MMMECH]{MediumMarket} on the set of agents $\N$ with budget $B$.
\STATE With probability $0.8$, execute \hyperref[alg:LMMECH]{LM-Mechanism} on the set of agents $\N$ with budget $B$.
\end{algorithmic}
\end{algorithm} 

\begin{algorithm}[t]
\caption{MediumMarket}\label{alg:MMMECH}
\begin{algorithmic}[1]
\STATE \textbf{Input}: Set of agents $\mathcal{N}$, budget $B$
\STATE Learn $\vmax$ using \hyperref[alg:Vmax]{LearningMaxValue}

\STATE Pick uniformly at random a threshold $t \in \left\{2^6 \cdot \vmax,2^8 \cdot \vmax,\dots,2^{23} \cdot \vmax\right\}$

\STATE Apply linear pricing with threshold $t$ and budget $B$ to the rest of the agent sequence.

\end{algorithmic}
\end{algorithm}

\begin{theorem}
    \hyperref[alg:PPMECH]{PostedPrices} is a universally truthful $O(1)$-competitive posted-price mechanism for online budget-feasible procurement auctions with secretary agent arrivals and  monotone submodular buyer's valuations.
\end{theorem}
\begin{proof}
    Universal truthfulness follows from the fact that 
    \hyperref[alg:PPMECH]{PostedPrices} is a probability distribution over three universally truthful (and posted-price) mechanisms. 
    
    Regarding the competitive ratio of \hyperref[alg:PPMECH]{PostedPrices}, we consider the following cases:
    \begin{enumerate}
        \item $\opt<1024\cdot \vmax$: Then Dynkin's algorithm is executed with probability $0.1$ and collects an expected value of at least $\frac{1}{10}\cdot\frac{\opt}{1024 \cdot e}$\,.
        \item $1024\cdot\vmax \le\opt<10^7\cdot \vmax$: Then \hyperref[alg:LMMECH]{MedianMarket} is executed with probability $0.1$. With probability at least $0.9$, the optimal value for the set of agents not used for the calculation of $\vmax$ by \hyperref[alg:Vmax]{LearningMaxValue}, in step~2, is at least $\opt / 4$. With probability at least $1/18$, the threshold $t$ selected is such that $\frac{\opt}{16}\le t \le \frac{\opt}{8}$\,. Thus, by Lemma~\ref{lem:threshold}, the chosen threshold $t$ results in at least $(\frac{1}{4}-\frac{\vmax}{\opt})\cdot \frac{\opt}{4}>\frac{\opt}{18}$ value. Overall the expected total value is at least:
        \begin{equation}
            \frac{1}{10}\cdot \frac{9}{10}\cdot \frac{1}{18} \cdot\frac{\opt}{18} = \frac{\opt}{3600}
        \end{equation}
        \item $10^7\cdot \vmax\le\opt$: Then, \hyperref[alg:LMMECH]{LM-Mechanism} is executed with probability $0.8$, which by Theorem \ref{theorem:lm} results in an expected total value of at least
        a $0.8\cdot\frac{1}{20}\cdot \frac{\opt}{4032e}$\,.
    \end{enumerate}
 Putting the three cases together, we conclude that the competitive ratio of \hyperref[alg:PPMECH]{PostedPrices} is $O(1)$. We note that at many places, we prioritized simplicity over trying to optimize the resulting constant.
\end{proof}

\section{Conclusions}

In this work, we have introduced a randomized, constant-competitive posted-price mechanism for online procurement auctions, thus resolving the main open question of Badanidiyuru, Kleinberg and Singer \cite{Bada2012}. Our findings demonstrate that in online procurement auctions, sequential posted-price mechanisms can be as powerful, in terms of the asymptotics of their competitive ratio, as seal-bid mechanisms, whose performance has been extensively studied. We note that despite our mechanism is elaborate to design and analyze, its interface with the agents is simple and transparent. 

In addition to improving the constants in our analysis, an interesting direction for further research is whether our refined adaptive search can be generalized and applied to other posted-price mechanisms, such as combinatorial auctions with submodular bidders \cite{Assadi019,AssadiS20,DuttingKL24}, towards  improved competitive ratios. 

\section{Acknowledgement}
This work has been partially supported by project MIS 5154714 of the National Recovery and Resilience Plan Greece 2.0 funded by the European Union under the NextGenerationEU Program.

\bibliographystyle{IEEEtran} 
\bibliography{references} 

\appendix

\clearpage\appendix
\section*{\Large Appendix}

\section{Detailed Proof of Theorem~\ref{thm:non-adaptive-lower-bound}}
\label{sec:app:lower_bound}


We consider a family of instances $\mathcal{I} = \{ I_0, I_1, ..., I_{\log n} \}$. Every instance consists of $n$ value-private cost pairs $(v,c)$, where for instance $I_i$, $v = 1$ and $c = B/2^i$, for every $i \in \{0,1,...,\log n\}$ (i.e., $I_i$ consists of $n$ value-cost pairs $(1, B/2^i)$).

We construct a probability distribution $\mathcal{D}_{\mathcal{I}}$ over the family $\mathcal{I}$ of instances such that the expected value of any deterministic non-adaptive linear-price mechanism on $\mathcal{D}_{\mathcal{I}}$ is at most $1$, while the expected value of the optimal solution on $\mathcal{D}_{\mathcal{I}}$ is $\frac{\log n}{2}+1$. Then, we apply Yao's principle, which extends the lower bound to any randomized mechanism (\cite[Chapter~8.4]{BoroYan1998} and \cite{Yao1977}). 

Specifically, for every $i = 0, 1, 2, \ldots \log(n)-1$, we let $I_{i}$ appear with probability $p_{i} = 1/2^{i+1}$ in $\mathcal{D}_{\mathcal{I}}$. We let $I_{\log n}$ appear with probability $p_{\log n} = 1/n$ in $\mathcal{D}_{\mathcal{I}}$, so that the sum of $p_i$'s is equal to $1$. The optimal solution in each $I_{i}$ is $2^i$. Therefore, the expected value of the optimal solution is 
\[ \Exp[\opt] = \frac{n}{n} + \sum_{i=0}^{\log(n)-1} \frac{2^i}{2^{i+1}}
=\frac{\log n}{2}+1\,.
\]

We let now fix any deterministic non-adaptive linear-price mechanism $\alg$. For simplicity and without loss of generality, we assume that $\alg$ is aware of $v$ in advance (clearly, this can only work in $\alg$'s favor). 
Since $\alg$ knows $v$ and does not get any information about the agents' private costs before it makes its first offer, $\alg$ can be regarded as selecting an arbitrary fixed index $i \in \{ 0, 1, \ldots, \log n\}$ and an arbitrary fixed threshold $\hat{t} \in [2^i, 2^{i+1})$. Then, $\alg$ determines its linear-price, which is the same for every agent. 

For simplicity, we let $\alg$'s fixed price be $B/2^i$, for some fixed $i \in \{ 0, 1, \ldots, \log n\}$. Then, $\alg$'s value for instance $I_j$ is $0$, if $j < i$, and $2^i$ for all $j = i, \ldots, \log n$. Therefore, the expected value of $\alg$'s solution is 
\[ \Exp[\alg] = \frac{2^i}{n} + 2^i \sum_{j=i}^{\log(n)-1} \frac{1}{2^{j+1}}
= \frac{2^i}{n} + 2^i \left(\frac{1}{2^i} - \frac{1}{n}\right) = 1\,.
\]
Then, the theorem follows from Yao's principle. \qed

\section{Proofs Missing from Section~\ref{sec:rounds}} 
\subsection{Partitioning into Rounds}\label{A1}
Suppose we wish to construct $\kappa$ rounds with length parameters $\{a_1,a_2,\dots,a_{\kappa}\}$. We will analytically compute the participation probability of the $\kappa$-th round.

\begin{proof}[Proof of Lemma \ref{lemma:RoundPro}:]
    Consider an agent $b$ and let $R_\kappa$ denote the set of agents in the $\kappa$-th round, then:
    \begin{equation}
        \Prob{b \in R_\kappa} = \sum_{x_{1} = 0}^{n_1}\Prob{x_{1}}\sum_{x_{2} = 0}^{n_{2}}\dots\sum_{x_{\kappa-1} = 0}^{n_{\kappa-1}}\Prob{x_{\kappa-1}}\sum_{x_{\kappa} = 1}^{n_{\kappa}}\Prob{x_{\kappa}}\cdot \Prob{b\in R_\kappa| x_1,\dots x_{\kappa}}.
    \end{equation}
    Where $n_1=n$ and $n_i = n-\sum_{j=1}^{i-1}x_i$. By using the fact that the order of the agents follows a uniform distribution  we get that:
    \begin{equation}
        \Prob{b\in R_\kappa} =\sum_{x_{1} = 0}^{n_1}\Prob{x_{1}}\sum_{x_{2} = 0}^{n_{2}}\dots\sum_{x_{\kappa-1} = 0}^{n_{\kappa-1}}\Prob{x_{\kappa-1}}\sum_{x_{\kappa} = 1}^{n_{\kappa}}\Prob{x_{\kappa}}\cdot \frac{x_{\kappa}}{n}.
    \end{equation}
     By definition, $\sum_{x_{\kappa} = 1}^{n_{\kappa}}\Prob{x_{\kappa}}\cdot x_{\kappa} = \E{}{\kappa\text{-th round length}} = a_{\kappa}\cdot n_{\kappa}$, thus:
    \begin{equation}
        \Prob{b\in R_\kappa} =\sum_{x_{1} = 0}^{n_1}\Prob{x_{1}}\sum_{x_{2} = 0}^{n_{2}}\dots\sum_{x_{\kappa-1} = 0}^{n_{\kappa-1}}\Prob{x_{\kappa-1}}\frac{n_{\kappa}}{n}\cdot a_{\kappa}.
    \end{equation}
     Observing that $n_{\kappa} = n_{\kappa-1} - x_{\kappa-1}$ and applying the same argument as before repeatedly, we finally get:
        \begin{equation}
        \Prob{b\in R_\kappa} =(1-a_1)\cdot(1-a_2)\cdots (1-a_{\kappa-1})\cdot a_{\kappa}.
    \end{equation}

    Now it remains to prove the independency among agents.

    Suppose 2 agents $b_1,b_2$, then:
    \begin{equation}
        \Prob{\{b_1,b_2\}\in R_\kappa} =\sum_{x_{1} = 0}^{n_1-2}\Prob{x_{1}}\sum_{x_{2} = 0}^{n_{2}-2}\dots\sum_{x_{\kappa-1} = 0}^{n_{\kappa-1}-2}\Prob{x_{\kappa-1}}\sum_{x_{\kappa} = 2}^{n_{\kappa}}\Prob{x_{\kappa}}\cdot \Prob{\{b_1,b_2\}\in R_\kappa|x_1,\dots,x_{\kappa}}.
    \end{equation}
    Since agent order follows a uniform distribution and the round lengths are drawn from a binomial distribution, we get:
    \begin{equation}
        \Prob{\{b_1,b_2\}\in R_\kappa} =\sum_{x_{1} = 0}^{n_1-2}\Prob{x_{1}}\sum_{x_{2} = 0}^{n_{2}-2}\dots\sum_{x_{\kappa-1} = 0}^{n_{\kappa-1}-2}\Prob{x_{\kappa-1}}\sum_{x_{\kappa} = 2}^{n_{\kappa}}\begin{pmatrix}
            n_\kappa\\
            x_\kappa
        \end{pmatrix}\cdot a_\kappa ^{x_\kappa}\cdot (1-a_\kappa)^{n_\kappa - x_\kappa}\cdot  \frac{\begin{pmatrix}
            x_\kappa\\
            2
        \end{pmatrix}}{\begin{pmatrix}
            n\\
            2
        \end{pmatrix}}.
    \end{equation}
    Using the identity: 
    \begin{equation}
        \begin{pmatrix}
            n_\kappa\\
            x_\kappa
        \end{pmatrix} \cdot \begin{pmatrix}
            x_\kappa\\
            2
        \end{pmatrix} = \begin{pmatrix}
            n_\kappa\\
            2
        \end{pmatrix}\cdot \begin{pmatrix}
            n_\kappa-2\\
            x_\kappa-2
        \end{pmatrix},
    \end{equation} we obtain:
    \begin{equation}
        \Prob{\{b_1,b_2\}\in R_\kappa} =\sum_{x_{1} = 0}^{n_1-2}\Prob{x_{1}}\sum_{x_{2} = 0}^{n_{2}-2}\dots\sum_{x_{1} = 0}^{n_{\kappa-1}-2}\Prob{x_{\kappa-1}}\cdot \frac{\begin{pmatrix}
            n_\kappa\\
            2
        \end{pmatrix}}{\begin{pmatrix}
            n\\
            2
        \end{pmatrix}}\sum_{x_{\kappa} = 2}^{n_{\kappa}}\begin{pmatrix}
            n_\kappa-2\\
            x_\kappa-2
        \end{pmatrix}\cdot a_\kappa ^{x_\kappa}\cdot (1-a_\kappa)^{n_\kappa - x_\kappa}.
    \end{equation}
    Noting that:
    \begin{equation}
        \sum_{x_{\kappa} = 2}^{n_{\kappa}}\begin{pmatrix}
            n_\kappa-2\\
            x_\kappa-2
        \end{pmatrix}\cdot a_\kappa ^{x_\kappa}\cdot (1-a_\kappa)^{n_\kappa - x_\kappa} = a_\kappa ^2 \cdot \sum_{x_{\kappa} = 0}^{n_{\kappa}-2}\begin{pmatrix}
            n_\kappa-2\\
            x_\kappa
        \end{pmatrix}\cdot a_\kappa ^{x_\kappa}\cdot (1-a_\kappa)^{n_\kappa-2 - x_\kappa} =a_\kappa ^2.
    \end{equation}
    and thus:
    \begin{equation}
        \Prob{\{b_1,b_2\}\in R_\kappa} =\sum_{x_{1} = 0}^{n_1-2}\Prob{x_{1}}\sum_{x_{2} = 0}^{n_{2}-2}\dots\sum_{x_{\kappa-1} = 0}^{n_{\kappa-1}-2}\Prob{x_{\kappa-1}}\cdot \frac{\begin{pmatrix}
            n_\kappa\\
            2
        \end{pmatrix}}{\begin{pmatrix}
            n\\
            2
        \end{pmatrix}}\cdot a_{\kappa}^2.
    \end{equation}
    Now we focus on the last sum,
    \begin{equation}\label{eq:2}
        \sum_{x_{\kappa-1} = 0}^{n_{\kappa-1}-2}\Prob{x_{\kappa-1}}\cdot \frac{\begin{pmatrix}
            n_{\kappa}\\
            2
        \end{pmatrix}}{\begin{pmatrix}
            n\\
            2
        \end{pmatrix}}\cdot a_{\kappa}^2 =a_{\kappa}^2\sum_{x_{\kappa-1} = 0}^{n_{\kappa-1}-2} \begin{pmatrix}
            n_{\kappa-1}\\
            x_\kappa
        \end{pmatrix}\cdot a_{\kappa-1} ^{x_{\kappa-1}}\cdot (1-a_{\kappa-1})^{n_{\kappa-1} - x_{\kappa-1}}\cdot \frac{\begin{pmatrix}
            n_\kappa\\
            2
        \end{pmatrix}}{\begin{pmatrix}
            n\\
            2
        \end{pmatrix}}.
    \end{equation}
    We notice that:
    \begin{equation}\label{eq:1}
        \begin{pmatrix}
            n_{\kappa-1}\\
            x_\kappa
        \end{pmatrix}\cdot \frac{\begin{pmatrix}
            n_\kappa\\
            2
        \end{pmatrix}}{\begin{pmatrix}
            n\\
            2
        \end{pmatrix}} = \begin{pmatrix}
            n_{\kappa-1}-2\\
            x_{\kappa-1}
        \end{pmatrix}\cdot \frac{\begin{pmatrix}
            n_{\kappa-1}\\
            2
        \end{pmatrix}}{\begin{pmatrix}
            n\\
            2
        \end{pmatrix}}, 
    \end{equation}
    and by applying  (\ref{eq:1}) to  (\ref{eq:2}) we get:
    \begin{equation}\label{eq:4}
        \sum_{x_{\kappa-1} = 0}^{n_{\kappa-1}-2}\Prob{x_{\kappa-1}}\cdot \frac{\begin{pmatrix}
            n_\kappa\\
            2
        \end{pmatrix}}{\begin{pmatrix}
            n\\
            2
        \end{pmatrix}}\cdot a_{\kappa}^2 =a_{\kappa}^2\frac{\begin{pmatrix}
            n_{\kappa-1}\\
            2
        \end{pmatrix}}{\begin{pmatrix}
            n\\
            2
        \end{pmatrix}}\cdot \sum_{x_{\kappa-1} = 0}^{n_{\kappa-1}-2}\begin{pmatrix}
            n_{\kappa-1}-2\\
            x_{\kappa-1}
        \end{pmatrix}\cdot a_{\kappa-1} ^{x_{\kappa-1}}\cdot (1-a_{\kappa-1})^{n_{\kappa-1} - x_\kappa}.
    \end{equation}
    Next we notice that:
    \begin{equation}\label{eq:3}
        \sum_{x_{\kappa-1} = 0}^{n_{\kappa-1}-2}\begin{pmatrix}
            n_{\kappa-1}-2\\
            x_{\kappa-1}
        \end{pmatrix}\cdot a_{\kappa-1} ^{x_{\kappa-1}}\cdot (1-a_{\kappa-1})^{n_{\kappa-1} - x_\kappa} = (1-a_{\kappa -1})^{2},
    \end{equation} and combining (\ref{eq:4}), (\ref{eq:3}) we get that:
    \begin{equation}
        \sum_{x_{\kappa-1} = 0}^{n_{\kappa-1}-2}\Prob{x_{\kappa-1}}\cdot \frac{\begin{pmatrix}
            n_\kappa\\
            2
        \end{pmatrix}}{\begin{pmatrix}
            n\\
            2
        \end{pmatrix}}\cdot a_{\kappa}^2 =\frac{\begin{pmatrix}
            n_{\kappa-1}\\
            2
        \end{pmatrix}}{\begin{pmatrix}
            n\\
            2
        \end{pmatrix}}\cdot (1-a_{\kappa-1})^2\cdot a_{\kappa}^2
    \end{equation}
    It is now evident that by following the same procedure for the rest of the sums we get that:
    \begin{equation}
        \Prob{\{b_1,b_2\}\in R_\kappa} = a_{\kappa}^2\cdot\prod_{i=1}^{\kappa-1}(1-a_i)^2 = \Prob{b_1\in R_\kappa}\cdot\Prob{b_2\in R_\kappa} 
    \end{equation}
    This confirms that agent selections are independent.
\end{proof}

Lemma \ref{lemma:RoundPro} also proves that rounds are indeed equivalent with a random set of the whole input where each agent is included with the same probability.\\

An important property of the lengths drawn in Period 2 that will turn our very useful in the rest of our analysis is the following:
\begin{lemma}\label{lemma:SeqBound}
    For every $i\in \{2,\dots, T\}$ (where $T$ is the number of intervals in the partitioning of $[\vmax, n\cdot\vmax]$ used in Section~\ref{sec:SecondPeriod}), it holds that:
    \[a_i\cdot \gamma_i \le \frac{1}{i^{10}}\]
\end{lemma}
\begin{proof}
    We first notice the following:
    \begin{equation}
        \begin{split}
            a_i \cdot \gamma_i = 81e\cdot \frac{\vmax}{t_i}\cdot \frac{3}{2}\cdot\log\left(\frac{t_i}{\vmax}\right)
        \end{split}
    \end{equation}
    Applying the fact that $\log\left(\frac{t_i}{\vmax}\right) < \frac{t_{i-1}}{\vmax}$, we get:
    \begin{equation}
         a_i \cdot \gamma_i < 122e\cdot \frac{t_{i-1}}{t_i}
    \end{equation}
    Now consider the following sequences of real numbers:
    \begin{equation}
        \begin{split}
        h_{1,i} &= 122e\cdot \frac{t_{i-1}}{t_i}\\
        h_{2,i} &= \frac{1}{i^{10}}
        \end{split}
    \end{equation}
    we have that:
    \begin{equation}
      h_{1,2} = \frac{122e}{10^7} < \frac{1}{1024} = h_{2,2}  
    \end{equation}
    Taking into account that $h_{1,i} = 122e\cdot \frac{t_{i-1}}{t_i} = 122e\cdot 2^{-\frac{t_{i-1}}{\vmax}}$ it becomes evident that $h_1$ drops to $0$ faster than $h_2$ and thus the desired inequality is proved.
\end{proof}

\subsection{The Proof of Lemma~\ref{lemma:PART}}
\label{sec:PART}

We are now ready to prove Lemma~\ref{lemma:PART}
\begin{proof}[Proof of Lemma \ref{lemma:PART}:]\label{proof:PART}
    First let $R$ be the set of all the rounds that \hyperref[alg:LMMECH]{LM-Mechanism} uses.
    Using lemma \ref{lemma:RoundPro} we get the following trivial bound for the participation probability $q_j$ of the $j-$th round:
    \begin{equation}
        q_j \ge a_j\cdot \prod_{i\in R}(1-a_i) 
    \end{equation}
    All that remains is to bound $\prod_{i\in R}(1-a_i)$.
    We break this product into 3 parts:
    \begin{enumerate}
        \item The first term of the product is $\frac{2}{3}$ and it corresponds to $\tau$.
        \item Next comes the following product of terms from the tests of the second period:
        \begin{equation}
        \begin{split}
        \prod_{i=2}^{T}(1-a_i)^{\gamma_i} \ge \prod_{i=2}^{T} (1-a_i\cdot \gamma_i), 
        \end{split}
        \end{equation}
        where the last inequality comes from taking the first order Taylor approximation of the convex function $g(x) = (1-x)^{\log(\frac{c}{x})}$.
        Using Lemma \ref{lemma:SeqBound}, we get:
        \begin{equation}
           \prod_{i=2}^{T}(1-a_i)^{\gamma_i} \ge \prod_{i=2}^{T} \left(1-\frac{1}{i^{10}}\right) \ge \frac{9}{10}.
        \end{equation} 
        \item Lastly, we have the terms from the BinarySearch and Exploitation:
        \begin{equation}
            \prod_{i=1}^{2\cdot\ell} (1-a) = (1-a)^{2\cdot\ell}\ge \frac{1}{e}.
        \end{equation}
    \end{enumerate}
    Combining all of the above together:
    \begin{equation}
        q_j \ge a_j\cdot \prod_{i\in R}(1-a_i) \ge \frac{6}{10e}\cdot a_j
    \end{equation}
\end{proof}
Before we delve deeper into the properties of our rounds we first need to establish a basic property of Threshold-Based mechanisms.

\subsection{The Proof of Lemma \ref{lemma:roundProb}}\label{proof4.5}

\begin{lemma}[\cite{Bada2012}]\label{lem:threshold}
Let $k = \frac{\opt}{v_{\max}}$. For any $ \hat{t} = h \cdot \opt$, such that $h\in (0,1)$, the following holds:
    \begin{equation}
            \text{PostedPrices}(\hat{t},\mathcal{N},B) \ge \min\left\{(1-h)\cdot \opt,\ \  \left(h-\frac{1}{k}\right)\cdot \opt\right\}\,,
        \end{equation} where $\text{PostedPrices}(\hat{t},\mathcal{N},B)$ denotes the expected value gathered by applying linear pricing with threshold $\hat{t}$ to the set $\N$ of agents with budget $B$.
\end{lemma}
\begin{proof}
    Suppose that we use \(\hat{t}=h \cdot \opt\) for the entire input. Let \( S \) be the allocation obtained using \(\hat{t}\).
    We consider the following cases:
    \begin{enumerate}
        \item There is an agent $b$ such that $b \in S^{*}\setminus S$ cause of insufficient Budget. In this case:
        \begin{equation}
            \begin{split}
                &\sum_{b_{j}\in S}p_{j} = \sum_{b_{j}\in S} \frac{B}{\hat{t}}(f_{S_{j}}(b_{j})) = \frac{B\cdot f(S)}{\hat{t}} = B\cdot \frac{f(S)}{\opt \cdot h} \text{ and } \\
            &c(b) \le \frac{B}{\hat{t}} \cdot f(b) \le B\cdot \frac{1}{\opt\cdot h} \cdot \frac{\opt}{k} = B\cdot \frac{1}{h\cdot k}
            \end{split}
        \end{equation}
        and thus
        \begin{equation}\begin{split}
            &B < c(b) + \sum_{b_{j}\in S}p_{j} \le B\cdot \frac{f(S)}{\opt \cdot h} + B\cdot \frac{1}{h\cdot k} \Longrightarrow \\
            &\opt\cdot \left(h-\frac{1}{k}\right) < f(S)
        \end{split}
        \end{equation}
        \item Otherwise: $c(b)> \frac{B}{\hat{t}}\cdot f_{S}(b)$ for every $b\in S^{*}\setminus S$ and thus
        \begin{equation}
            \sum_{b\in S^{*}\setminus S} f_S(b) = \sum_{b\in S^{*}\setminus S} \frac{f_{S}(b)}{c(b)} \cdot c(b) \le \sum_{b\in S^{*}\setminus S}\frac{\hat{t}}{B}\cdot c(b)\le \hat{t}  
        \end{equation}
        and it follows that \begin{equation}
        \begin{split}
                        &\opt-f(S) \le f(S^{*}\cup S) -f(S) \le \sum_{b\in S^{*}\setminus S} f_S(b)\le \hat{t} \Longrightarrow \\
            &f(S) \ge (1-h)\cdot \opt
        \end{split}
        \end{equation}
    \end{enumerate}
    
    Finally, we conclude that $f(S)\ge \min\{(1-h),(h-\frac{1}{k})\}\cdot \opt$. 
\end{proof}

We notice that the last bound is maximized for $\hat{h} = \frac{k+1}{2k}$ and thus we consider it to be exactly that for the rest of our analysis.

Let $\hat{t}=\hat{h}\cdot \opt = \frac{k+1}{2k}\cdot \opt$ and $S$ be the allocation we get by applying the threshold $\hat{t}$ to the whole input, for some permutation of the agents. Denote $\{s_i\}_{i=1}^{|S|}$ the elements of $S$ in order of appearance and let $S_i = \bigcup_{j=1}^{i}\{s_j\}$. Now, letting $w_i = f_{S_{i-1}}(s_i)$, we can bound the cost of all subsets $X\subseteq S$ the following way:
 
\begin{lemma}\label{lem4.6}
     For every subset $X\subseteq S$ it is true that:
   \begin{equation}
       \begin{split}
           c(X)\le \frac{\sum_{i: s_i \in X}w_i}{\hat{h}\cdot \opt}\cdot B
       \end{split}
   \end{equation}
\end{lemma}
\begin{proof}
    The cost of each agent in $S$ is at most $\frac{f_{S_i}(s_i)}{\hat{t}}\cdot B = \frac{w_i}{\hat{t}}\cdot B.$ Consequently, for each subset $X$ of $S$ we get: \begin{equation}
        \begin{split}
            c(X) = \sum_{j:b_j \in X}c(b_j) 
            \le \sum_{i:s_i \in X}\frac{w_i}{\hat{t}}\cdot B . 
        \end{split}
    \end{equation}
\end{proof}
The following lemma showcases a way to lower bound the value of subsets $X \subseteq S$.
\begin{lemma}\label{lem4.7}
    Let $I$ be the set of indices of elements $s_i\in S$ sampled in a round. Then, the following holds:
    \begin{equation}
        \begin{split}
            f\lp(\bigcup_{j\in I}\{s_j\}\rp) \ge \sum_{j\in I} w_j
        \end{split}
    \end{equation}
\end{lemma}
\begin{proof}
   Let $J = \bigcup_{j\in I}\{s_j\}$. The value of the set $J$ is the same regardless of the order of appearance of elements $s_j$. Thus we can suppose that they appear in the same order as in $S$ and thus:
   \begin{equation}
       \begin{split}
           f(J) = \sum_{j\in I} f_{S_{j-1}\cap J}(s_j) \ge \sum_{j\in I} f_{S_{j-1}}(s_j) = \sum_{j\in I} w_j  
       \end{split}
   \end{equation}
\end{proof}

Combining the lemmas \ref{lem4.6} and \ref{lem4.7} we get the following bound for the cost of all subsets $X\subseteq S$.

\begin{corollary}\label{cor4.2}
    For every $X\subseteq S$ it holds that:
    \begin{equation}
        \begin{split}
            c(X) \le \frac{2k}{k+1}\cdot \frac{f(X)}{OPT}\cdot B
        \end{split}
    \end{equation}
\end{corollary}
Let, once again, $R$ be the set of rounds. Suppose that $X$ is the subset of agents of $S$ that is included in a round $j$ with length parameter $a_j$, such that $a_j \cdot \frac{\opt}{\vmax}\ge81\, e$. Now
consider the following random variables for $j\in R, i \in [|S|]$:
\begin{equation}
    \begin{split}
            X_i^{j} &= \begin{cases} 
      w_i & \text{ w.p.  } \frac{6}{10e}\cdot a_j\\
      0 & \text{ otherwise  } 
   \end{cases}\\
   X^{j}&= \sum_{i=1}^{|S|}X_{i}^{j}
    \end{split}
\end{equation}

 By lemmas ~\ref{lemma:RoundPro},~\ref{lem4.7} it is evident that $X^{j}$ is a lower bound of the value $f(X)$ contained in round $j$. Furthermore, we can prove the following:

\begin{lemma}\label{lemmaTool}
    Let $\opt_j$ be the value of an optimal solution for an $a_j$-round $j$, when restricted to budget $B_j = 3\cdot C\cdot a_j\cdot B$ budget, where $a_j$ is good with respect to $\opt$. If $C\cdot a_j\cdot \opt\le X^j$ then 
    \begin{equation}
        \begin{split}
            \opt_j \ge C\cdot a_j\cdot \opt
        \end{split}
    \end{equation}
\end{lemma}
\begin{proof}
Let $X$ be the set of agents included in round $j$. Define $S_{j} = X\cap S$ and denote by $s_{j,i}$ the $i-$th element of $S_j$ in the order of appearance within the round. We can select the first $\mu$ agents $ \{s_{j,i}\}_{i=1}^{\mu}$ (for some $\mu \in \mathbb{N}$) such that their total value $f( \{s_{j,i}\}_{i=1}^{\mu})$ does not exceed $C\cdot a_j\cdot \opt$ by more than $\vmax.$ More specifically, there exists $\mu \in \mathbb{N}$ such that:
\begin{equation}
    \begin{split}
        C\cdot a_j\cdot \opt &\le f\left(\bigcup_{i=1}^{\mu}\{s_{j,i}\}\right ) \\
        &\le C\cdot a_j\cdot \opt+ \vmax\ \\
        &= \left(C\cdot a_j+ \frac{\vmax}{\opt}\right)\cdot \opt 
    \end{split}
\end{equation}
Which, by Corollary \ref{cor4.2}, implies:
\begin{equation}
            c\lp(\bigcup_{i=1}^{\mu}\{s_{j,i}\}\rp) \le 2\cdot\left (C\cdot a_j \cdot + \frac{\vmax}{\opt}\right)\cdot B\le 3\cdot C\cdot a_j\cdot B,
\end{equation}
where the last inequality is due to $a_j$ being good with respect to $\opt$, i.e., $a_j \geq 81\,e\,\vmax\,/\,\opt$.
\end{proof}

Lemma~\ref{lemmaTool} enables us to analyze the value of $X^j$ to determine whether a round is dense, rather than relying on $\opt_j$. Now we can use concentration bounds on the random variable $X_{j}$ to bound the probability that the $j$-th round is dense.

We now ready to conclude this section with presenting the proof of Lemma~\ref{lemma:roundProb}.

\begin{proof}[Proof of lemma \ref{lemma:roundProb}:]

     By Lemma \ref{lemmaTool}:
    \begin{equation}
        \begin{split}
            \Prob{\text{round $j$ is not dense}} \le \Prob{X^{j}<C\cdot a_j \cdot \opt}
        \end{split}
    \end{equation}
    We notice that $\frac{1}{2}\cdot \frac{6}{10e} \cdot \frac{k-1}{2k} \ge C$ (for $k\ge 10^7$) and thus:
    \begin{equation}
        \Prob{\text{round $j$ is not dense}} \le \Prob{X^{j}<\frac{1}{2}\cdot \frac{6}{10e}\cdot a_j \cdot \frac{k-1}{2k}\cdot \opt}
    \end{equation}
    Now using the fact that $\E{}{X^j} = q_j\cdot f(S)$, along with Lemma \ref{lemma:RoundPro}:
        \begin{equation}
        \Prob{\text{round $j$ is not dense}} \le\Prob{X^{j}\le \frac{1}{2}\cdot \E{}{X^{j}}} 
    \end{equation}
    We can now apply Chernoff bound (setting $\delta = \frac{1}{2}$) on random variable $X^{j}$ to get:
    \begin{equation}
        \begin{split}
           \Prob{\text{round $j$ is not dense}} &\le \Prob{X^{j}\le (1-\delta)\E{}{X^{j}}}\\
            &\le \exp\lp\{\frac{-\delta^2 \E{}{X^{j}}}{2\cdot \max{X_{i}^{j}}}\rp\}
        \end{split}
    \end{equation}
    Using the fact that $\E{}{X^{j}} =  \frac{6}{10e}\cdot a_j\cdot f(S) \ge \frac{6}{10e}\cdot a_j\cdot \frac{k-1}{2k}\cdot \opt$:
    \begin{equation}
        \Prob{\text{round $j$ is not dense}} \le \exp\lp\{\frac{-\delta^2 \cdot \frac{6}{10e}\cdot a_j\cdot\frac{k-1}{2k}\cdot \opt}{2\vmax}\rp\} 
    \end{equation}
    After calculations we reach the following:
    \begin{equation}
        \Prob{\text{round $j$ is not dense}} \le\exp\lp\{\frac{- a_j \cdot  (k-1)}{27e}\rp\} \le 0.1
    \end{equation}
    The last inequality follows from the fact that $a_j$ is good with respect to $\opt$, i.e., $a_j \geq 81\,e\,\vmax\,/\,\opt$, and from the hypothesis that $k\ge 10^7$.
\end{proof}

\subsection{The Proof of Lemma~\ref{lemma:PhaseProbability}}\label{negdep}

To extend our analysis to dense phases we should first pay attention to the dependency between the probability of success of consecutive rounds. Indeed, we should notice that conditional to the previous round having failed the next round is more likely to succeed. Such a dependency bares a significant resemblance to the Balls and Bins problem \cite{johnson1977urn}. Below we present how to utilize results from \cite{negativedepend} to prove that concentration of mass arguments still hold in this setting.

In the following we will present a complete proof of the validity of Chernoff bound on the following random variables:

\begin{equation}
    \begin{split}
    Q_i &= \begin{cases} 
      1 ,& \text{ if round } $i$ \text{ is dense}\\
      0, & \text{ otherwise  } 
   \end{cases}\\
    \end{split}
\end{equation}

Firstly, we need to present some important notions from \cite{negativedepend}.

\begin{definition}[Negative Association]
   Let $\mathbf{X} := (X_1, \dots, X_n)$ be a vector of random variables.

\begin{enumerate}
    \item [$(-A)$] The random variables, $\mathbf{X}$ are \emph{negatively associated} if for any two disjoint index sets, $I, J \subseteq [n]$,
\[
\mathbb{E}[f(X_i, i \in I)g(X_j, j \in J)] \leq \mathbb{E}[f(X_i, i \in I)]\mathbb{E}[g(X_j, j \in J)]
\]
for all functions $f : \mathbb{R}^{|I|} \to \mathbb{R}$ and $g : \mathbb{R}^{|J|} \to \mathbb{R}$ that are both non-decreasing or both non-increasing.
\end{enumerate}
\end{definition}

\begin{lemma}[Zero-One Lemma for ($-A$)] \label{ZeroOne}
    If $X_1, \dots, X_n$ are zero-one random variables such that $\sum_i X_i = 1$, then $X_1, \dots, X_n$ satisfy ($-A$).
\end{lemma}

\begin{proposition}\label{proptrans}
    \begin{enumerate}
    \item If $\mathbf{X}$ and $\mathbf{Y}$ satisfy ($-A$) and are mutually independent, then the augmented vector $(\mathbf{X}, \mathbf{Y}) = (X_1, \dots, X_n, Y_1, \dots, Y_m)$ satisfies ($-A$).
    
    \item Let $\mathbf{X} := (X_1, \dots, X_n)$ satisfy ($-A$). Let $I_1, \dots, I_k \subseteq [n]$ be disjoint index sets, for some positive integer $k$. For $j \in [k]$, let $h_j : \mathbb{R}^{|I_j|} \to \mathbb{R}$ be functions that are all non-decreasing or all non-increasing, and define $Y_j := h_j(X_i, i \in I_j)$. Then the vector $\mathbf{Y} := (Y_1, \dots, Y_k)$ also satisfies ($-A$). That is, non-decreasing (or non-increasing) functions of disjoint subsets of negatively associated variables are also negatively associated.
\end{enumerate}
\end{proposition}

\begin{proposition}\label{lemChern}
    The Chernoff--Hoeffding bounds are applicable to sums of variables that satisfy the negative association condition ($-A$).
\end{proposition}

 We are now ready to complete our proof:
 \begin{lemma}
     We can apply Chernoff bound on $\{Q_i\}_{i\in[\ell]}$
 \end{lemma}
 \begin{proof}
     Consider the following random variables:
     \begin{equation}
         \begin{split}
              Y_i^{j} &= \begin{cases} 
      1 & \text{ if agent $b_i$ is in round $j$}\\
      0 & \text{ otherwise  } 
   \end{cases}\\
    Q_j &= \begin{cases} 
      1 & \text{ if  } \sum_{i=1}^{n}w_i \cdot Y_{i}^{j} \ge  C\cdot a_j \cdot  OPT\\
      0 & \text{ otherwise  } 
   \end{cases}
         \end{split}
     \end{equation}

     By lemma \ref{ZeroOne} it is immediate that $(Y_i^{1},Y_{i}^2,\dots,Y_{i}^{\ell})$ are negatively associated. By Proposition \ref{proptrans}.1, along with the independency among agents, we can conclude that the whole set $(Y_{i}^{j})_{i\in [n],j\in[\ell]}$ is negatively associated. Finally, we use proposition \ref{proptrans}.2, setting 
     \begin{equation}
         \begin{split}
             h_j = \begin{cases} 
      1 & \text{ if  } \sum_{i=1}^{n}w_i\cdot Y_{i}^{j} \ge C\cdot a\cdot OPT\\
      0 & \text{ otherwise  } 
   \end{cases}
         \end{split}
     \end{equation} which is an increasing function of $Y_i^{j}$, to get that
     $(Q_1,Q_2,\dots,Q_\ell)$ are negatively associated. We now get the desired result directly from Proposition \ref{lemChern}
 \end{proof}

The following lemma effectively leverages this dependency to demonstrate that all phases, whose rounds are good with respect to their thresholds, are dense with significant probability.

\begin{proof}[Proof of Lemma \ref{lemma:PhaseProbability}]
Proposition \ref{lemma:roundProb} suggests that every round succeeds with probability at least $0.9$. Now let $\{Q_i\}_{i=1}^{\delta}$ be random variables under the following distribution:
\begin{equation}
    \begin{split}
    Q_i &= \begin{cases} 
      1 ,& \text{ if round } $i$ \text{ is dense}\\
      0, & \text{ otherwise  } 
   \end{cases}\\
    \end{split}
\end{equation}
 By leveraging the previous results, we can argue that the variables $ Q_i $ exhibit negative dependence, which permits the use of concentration inequalities.

Let $Q = \sum_{i=1}^m  Q_i$. The probability a Phase is not dense is, by applying Chernoff bound on $Q$, at most:
\begin{equation}
    \begin{split}
       \Prob{ Q < \frac{\delta}{2} } &\le \Prob{ Q < \frac{1}{1.8} \cdot \E{}{Q} } \\
       &\le \exp\lp\{ - \frac{(\frac{0.8}{1.8})^2\cdot 0.9 \cdot \delta}{2} \rp\} \\
       & = \exp\lp\{ - \frac{4 \cdot \delta}{45} \rp\}
    \end{split}
\end{equation}
\end{proof}

 \subsection{Lemmas Used in Estimating the Correctness Probability in Section~\ref{sec:estimate}}
 \label{A3}

\begin{lemma}\label{lemma:E2}
    The event $\mathcal{E}_2$ happens with probability at least $0.9$.
\end{lemma}
\begin{proof}
    To successfully identify the correct intervals, we need the last, or the second to last, successful Phase to include $\frac{\opt}{8}$. Let us denote by $P_i$ the Phase that tests threshold $t_i$. Let $\frac{\opt}{8}\in (t_j,t_{j+1}]$, then, if $P_j$ is dense, $P_j$ will be successful and all $P_i$, with $i>j+1$ will fail, by Lemma \ref{lemma:PERIOD1}. Thus, the set of intervals we end up with will include $(t_j,t_{j+1}].$ Thus, we need to investigate the probability that Phase $P_j$ is dense. By Lemma \ref{lemma:PhaseProbability} we have that:
\begin{equation}
    \begin{split}
            \Prob{P_j \text{ is dense}} &\ge 1-\exp\left\{-\frac{4\cdot \gamma_{j}}{45}\right\}\\
        &\ge 1-\exp\left\{-\frac{4\cdot \gamma_{2}}{45}\right\}  =1-\exp\left\{-3\right\}\ge 0.9 
    \end{split}
\end{equation}
\end{proof}

 \begin{lemma}\label{lemma:BinaryProb}
    All phases tested with a threshold $t\le\frac{\opt}{8}$, during BinarySearch and Exploitation, are dense with probability at least $0.9$ (i.e., the event $\Event_3$ occurs with probability at least $0.9$). 
\end{lemma}

\begin{proof}
In the worst case we might need all of the $\frac{2\ell}{m}$ Phases to be dense.
Lemma \ref{lemma:roundProb} suggests that every round succeeds with probability at least $0.9$. Now let $\{Q_i\}_{i=1}^{m}$ be random variables under the following distribution:
\begin{equation}
    \begin{split}
    Q_i &= \begin{cases} 
      1 ,& \text{ if round } $i$ \text{ is dense}\\
      0, & \text{ otherwise  } 
   \end{cases}\\
    \end{split}
\end{equation}

Let $Q = \sum_{i=1}^m  Q_i$.
As argued before (\ref{negdep}), we can apply concentration of mass arguments to $Q$. The probability a Phase is not dense is, by applying Chernoff bound on $Q$, at most:
\begin{equation}
           \Prob{ Q < \frac{m}{2} } \le \Prob{ Q < \frac{1}{1.8} \cdot \E{}{Q} } 
       \le \exp\lp\{ - \frac{(\frac{0.8}{1.8})^2\cdot 0.9 \cdot m}{2} \rp\}
\end{equation}
Using the definitions of $m$ and $t_{\max}$, we get:
\begin{equation}
    \Prob{ Q < \frac{m}{2} } \le\frac{1}{\log\left(\frac{t_{\max}}{t_{\min}}\right)}= \frac{1}{t_{\min}}
\end{equation}
Leveraging our hypothesis that $t_{\min}\ge 10^7$:
\begin{equation}\label{eq:77}
    \Prob{ Q < \frac{m}{2} } \le \frac{1}{20\left\lceil\log\left(t_{\min}\right)\right\rceil}
\end{equation}
To get the desired probability bound, we will now use Union Bound on the event where the first $\frac{2\cdot\ell}{m}$ Phases are dense. Namely:
\begin{equation}
    \begin{split}
         \Prob{\text{First }\frac{2\cdot\ell}{m} \text{ Phases are dense} } &\ge 1 - 2\cdot \left\lceil\log\log\left(\frac{t_{\max}}{t_{\min}}\right)\right\rceil \cdot \Prob{\text{A Phase is not dense}} 
    \end{split}
\end{equation}
Now using the definition of $t_{\max} = 2^{\frac{t_{\min}}{\vmax}}\cdot t_{\min}$, along with inequality~\ref{eq:77}:
\begin{equation}
    \Prob{\text{First }\frac{2\cdot\ell}{m} \text{ Phases are dense} } \ge 1 - \frac{\left\lceil\log\left(t_{\min}\right)\right\rceil}{10\cdot \left\lceil\log\left(t_{\min}\right)\right\rceil} = 0.9
\end{equation}
\end{proof}

Sampling round lengths from a distribution introduces the risk that we run out of agents during the execution of our mechanism. Consequently, our analysis relies on the assumption that the total round lengths remain within the input size. Fortunately, the following lemma guarantees that this condition holds with constant probability:  

\begin{lemma}\label{lemma:Fit}
    We will not run out of agents during the execution of \hyperref[alg:LMMECH]{LM-Mechanism} with probability at least $0.97$ (i.e., the event $\Event_4$ occurs with probability at least $0.97$).
\end{lemma}
\begin{proof}
    We will investigate the round lengths drawn as follows:
\begin{enumerate}
    \item Applying Chernoff bound onto the random length $\tau$ drawn for the estimation of $\vmax$ we get that:
    \begin{equation}
        \begin{split}
            \Prob{\tau \ge \frac{2}{5}\cdot n} \le \exp\left\{-\frac{n}{225}\right\}\le 0.01,
        \end{split}
    \end{equation}
    where the last inequality holds for $n\ge10^7$, which must be true for our assumption $\opt\ge10^7\vmax$ to hold.
    Thus:
    \begin{equation}
        \Prob{\tau < \frac{2}{5}\cdot n}\ge 0.99
    \end{equation}
    \item We now consider the random lengths $n_{i,j}$ drawn during Period 2. Consider $m_{i,j}\sim B(n,a_{i})$. If we were to use $m_i$ as our round lengths then our risk of running out of agents would be greater than with $n_i$. We observe that: 
    \begin{equation}
        \E{}{\sum_{i=2}^{T}\sum_{j=1}^{\gamma_i} n_{i,j}}\le\E{}{\sum_{i=2}^{T}\sum_{j=1}^{\gamma_i} m_{i,j}} = n\cdot \sum_{i=2}^{T}\gamma_i \cdot a_i
    \end{equation}
    By applying Markov's inequality on the sum of these random variables, and noticing that  $\E{}{m_{i,j}} = n\cdot a_{i}$, we get:
    \begin{equation}
        \Prob{\sum_{i=2}^{T}\sum_{j=1}^{\gamma_i} n_{i,j} \ge 100\cdot n \cdot\sum_{i=2}^{T} a_{i} \cdot \gamma_i}\le\Prob{\sum_{i=2}^{T}\sum_{j=1}^{\gamma_i} m_{i,j} \ge 100\cdot n \cdot\sum_{i=2}^{T} a_{i} \cdot \gamma_i}\le 0.01
    \end{equation}
    Thus:
    \begin{equation}
        \begin{split}
            \Prob{\sum_{i=2}^{T}\sum_{j=1}^{\gamma_i} n_{i,j} \ge 100\cdot n \cdot\sum_{i=2}^{T} a_{i} \cdot\gamma_i}\le 0.01 
        \end{split}
    \end{equation}
    By lemma \ref{lemma:SeqBound} we get  $\sum_{i=2}^{T}a_i \cdot \gamma_i < \frac{1}{500}$, and thus:
    \begin{equation}
        \begin{split}
            100\cdot n \cdot\sum_{i=2}^{T} a_{i} \cdot \gamma_i < \frac{n}{5}
        \end{split}
    \end{equation}
    Overall:
    \begin{equation}
        \Prob{\sum_{i=2}^{T}\sum_{j=1}^{\gamma_i} n_{i,j} < \frac{n}{5}}\ge 0.99 
    \end{equation}
    \item Finally, we draw $2\cdot\ell$ lengths $n_i\sim B(N_i,a)$ for the rounds of BinarySearch and Exploitation with the same length parameter $a$. Suppose, similarly with before, the random variables $m_i \sim B(n,a)$ for $i\in[2\,\ell]$. It is evident that 
    \begin{equation}
        \E{}{\sum_{i=1}^{2\,\ell}n_i}\le\E{}{\sum_{i=1}^{2\,\ell}m_i} =  n\cdot 2\,\ell \cdot a = \frac{n}{3}.
    \end{equation} 
    By applying now Chernoff bound on $m_i$ we get that:
        \begin{equation}
        \begin{split}
                        \Prob{\sum_{i=1}^{2\,\ell} n_{i} \ge \frac{2}{5}\cdot n}\le\Prob{\sum_{i=1}^{2\,\ell} m_{i} \ge \frac{2}{5}\cdot n}\le 0.01. 
        \end{split}
    \end{equation}
    Overall we get:
            \begin{equation}
        \begin{split}
\Prob{\sum_{i=1}^{2\,\ell} n_{i} < \frac{2}{5}\cdot n}\ge 0.99
        \end{split}
    \end{equation}
\end{enumerate}
All in all, by Union Bound on the events described before, we get that we will not run out of agents during our Mechanism with probability at least $1 - 0.01-0.01-0.01 = 0.97$ 
\end{proof}


\section{The Proof of Lemma~\ref{BUDGET}}\label{budgetproofs}

Below we prove that all needed rounds can be conducted without fear of running out of budget.
\begin{proof}[Proof of Lemma \ref{BUDGET}:]
According to Property \ref{property:a} during Period 2 in worst case the overall budget expended is the following:
\begin{equation}
    \begin{split}
        \sum_{i=2}^{T}  \gamma_i\cdot3\cdot C\cdot a_i\cdot B &\le\sum_{i=2}^{T}\frac{3\cdot C}{i^{10}}\cdot B < 0.04\cdot B ,
    \end{split}
\end{equation}
where the second to last inequality follows from Lemma~\ref{lemma:SeqBound}.

According to Property \ref{property:a}, during Period 3 and 4, in worst case the overall budget expended is the following:
\begin{equation}
    \begin{split}
        \sum_{i=1}^{2\cdot\ell} 3\cdot C\cdot a\cdot B         &\le 6\cdot C\cdot \ell \cdot a\cdot B \\
        &= C\cdot B\\
        &\le 0.06\cdot B
    \end{split}
\end{equation}    
Overall we have expended a budget of at most $B/10$.
\end{proof}

\end{document}